\documentclass[11pt]{article}
\pdfoutput=1
\usepackage{fullpage}

\usepackage{graphicx}   
\usepackage{hyperref}   
\usepackage{amsmath, amsthm,amssymb}    
\usepackage{mdwlist}
\usepackage{xspace}
\usepackage[usenames,dvipsnames]{color}
\usepackage{verbatim}   
\usepackage{subfig}
\usepackage{enumitem}

\usepackage[normalem]{ulem}

\newcommand{\red}[1]{{\color{red} {#1}}}
\newcommand{\cut}[1]{}   

\newenvironment{packed_item}{
\begin{itemize}
   \setlength{\itemsep}{1pt}
   \setlength{\parskip}{0pt}
   \setlength{\parsep}{0pt}
}
{\end{itemize}}

\newenvironment{packed_enum}{
\begin{enumerate}
   \setlength{\itemsep}{1pt}
  \setlength{\parskip}{0pt}
   \setlength{\parsep}{0pt}
}
{\end{enumerate}}

\newcommand{\ie}{{\em i.e.}\xspace}
\newcommand{\eg}{{\em e.g.}\xspace}
\newcommand{\ea}{{\em et al.}\xspace}

\newcommand{\introparagraph}[1]{\textbf{#1.}}  

\newcommand{\set}[1]{\{#1\}}                    
\newcommand{\setof}[2]{\{{#1}\mid{#2}\}}        

\usepackage{aliascnt} 

\newtheorem{theorem}{Theorem}[section]          	
\newaliascnt{lemma}{theorem}				
\newtheorem{lemma}[lemma]{Lemma}              	
\aliascntresetthe{lemma}  					
\newaliascnt{conjecture}{theorem}			
\aliascntresetthe{conjecture}  				
\newaliascnt{remark}{theorem}				
\newtheorem{remark}[remark]{Remark}              
\aliascntresetthe{remark}  					
\newaliascnt{fact}{theorem}				
              
\aliascntresetthe{fact}  					
\newaliascnt{corollary}{theorem}			
\aliascntresetthe{corollary}  				
\newaliascnt{definition}{theorem}			
\newtheorem{definition}[definition]{Definition}    
\aliascntresetthe{definition}  				
\newaliascnt{proposition}{theorem}			
\newtheorem{proposition}[proposition]{Proposition}  
\aliascntresetthe{proposition}  				
\newaliascnt{observation}{theorem}				
\newtheorem{observation}[observation]{Observation}  
\aliascntresetthe{observation}  				
\newaliascnt{example}{theorem}			
\newtheorem{example}[example]{Example}  	
\aliascntresetthe{example}  				

\providecommand{\att}[1]{\text{att}(#1)}

\providecommand{\ba}[0]{\mathbf{\alpha}}

\providecommand{\mI}[0]{\mathcal{I}}
\providecommand{\mH}[0]{\mathcal{H}}
\providecommand{\mP}[0]{\mathcal{P}}
\providecommand{\mE}[0]{\mathcal{E}}
\providecommand{\mQ}[0]{\mathcal{Q}}
\providecommand{\dom}[0]{\mathtt{Dom}}

\providecommand{\tup}[0]{\mathtt{tup}}

\title{Answering Conjunctive Queries with Inequalities\thanks{
This work has been partially funded by the NSF awards IIS-1247469 and IIS-0911036, European Research Council under the FP7, ERC grant MoDaS, agreement 291071 and by the Israel Ministry of Science.}}
\author{Paraschos Koutris \and Tova Milo \and Sudeepa Roy \and Dan Suciu}

\begin{document}

\maketitle

\begin{abstract}
In this paper, we study the complexity of answering conjunctive
queries (CQ) with inequalities ($\neq$).  In particular, we are
interested in comparing the complexity of the query with and without
inequalities. 
The main contribution of our work is a novel combinatorial technique that enables us to use 
any Select-Project-Join query plan for a given CQ without inequalities in 
answering the CQ with inequalities, with an additional factor in running time 
that only depends on the query. 
The key idea is to define a new projection operator, which keeps 
a small representation (independent of the size of the database) of 
the set of input tuples that map to each tuple in the output of the projection; this representation is 
used to evaluate all the inequalities in the query.
 Second, we generalize a result by Papadimitriou-Yannakakis \cite{PY97} 
 and give an alternative algorithm based on the color-coding technique \cite{AlonYZ08}
 to evaluate a CQ with inequalities 
 by using an algorithm for the CQ without inequalities.  
Third, we investigate the structure of the query graph, inequality graph, 
and the augmented query graph with inequalities, and show that 
even if the query and the inequality graphs have bounded treewidth, the augmented graph 
not only can have an unbounded treewidth but can also be NP-hard to evaluate.  
Further, we illustrate classes of queries and inequalities where the augmented graphs have unbounded treewidth, 
but the CQ with inequalities can be evaluated in poly-time.
Finally, we give necessary properties and sufficient properties that allow a class of CQs to have 
poly-time combined complexity with respect to any inequality pattern.
We also illustrate classes of queries where our query-plan-based technique 
outperforms the alternative approaches 
discussed in the paper. 

\end{abstract}

\section{Introduction}
\label{sec:intro}

In this paper, we study the complexity of answering conjunctive queries (CQ) with a set of inequalities of the form $x_i \neq x_j$ between variables in the query. 
The complexity of answering CQs without inequalities has been extensively studied in the literature during the past three decades. 
Query evaluation of CQs is NP-hard in terms of \emph{combined complexity} (both query and database are inputs), 
while
the  {\em data complexity} of CQs (query is fixed) is in $AC_0$ \cite{AbiteboulHV95}. Yannakakis~\cite{Yan81} showed  that evaluation of acyclic CQs has polynomial-time combined complexity. 
This result has been generalized later to CQs with bounded treewidth, bounded querywidth, or bounded hypertreewidth:  the combined complexity remains polynomial if the width of a 
tree or query decomposition of the query (hyper-)graph is bounded \cite{ChekuriR00, Gottlob1999, Kolaitis1998, Flum2002}.

However, the complexity of query evaluation changes drastically once we add inequalities in the body of the query. 
Consider the following Boolean acyclic CQ $P^k$ which can be solved in $O(k|D|)$ time on a database instance $D$:
$$P^k(~) = R_1(x_1, x_2), R_2(x_2, x_3), \dots, R_k(x_{k}, x_{k+1})$$
If we add the inequalities $x_i \neq x_j$ for every $i < j$ and evaluate it on an instance where each $R_{\ell}, 1 \leq \ell \leq k$, corresponds to the edges in a graph with $k+1$
vertices, query evaluation becomes equivalent to asking whether the graph contains a Hamiltonian path, and therefore is NP-hard in $k$.
Papadimitriou and Yannakakis~\cite{PY97} observed this fact and showed that still the problem is {\em fixed-parameter tractable} for acyclic CQs:
\begin{theorem}[\cite{PY97}]\label{thm:PY97_acyclic}
Let $q$ be an acyclic conjunctive query with inequalities and $D$ be a database instance. Then, $q$ can be evaluated in time $2^{O(k \log k)} \cdot |D| \log^2 |D|$
where $k$ is the number of variables in $q$ that appear in some inequality. 
\end{theorem}

The proof is based on the {\em color-coding} technique introduced by Alon-Yuster-Zwick in \cite{AlonYZ08} 
that finds subgraphs in a graph.
In general, answering CQs with inequalities is closely related to finding patterns in a graph, which has been extensively studied 
in the context of graph theory and algorithms.
For example, using the idea of \emph{representative sets},  Monien \cite{Monien85} showed the following: given a graph $G(V, E)$ and a vertex $s \in V$, there exists a deterministic $O(k! \cdot |E|)$ algorithm that finds all vertices 
$v$ with a length-$k$ path from $s$ and also reports these paths (a trivial algorithm will run in time $O(|V|^k)$).
Later, Alon \ea\ proposed the much simpler color-coding technique that can solve the same problem in expected time $2^{O(k)} |V|$ for 
undirected graphs and $2^{O(k)} |E|$ for directed graphs. These two ideas have been widely used to find other patterns in a graph, \eg, for finding 
cycles of even length \cite{AYZ97, YusterZ97, AlonYZ08}.

In the context of databases, 
 Papadimitriou and Yannakakis \cite{PY97} showed that 
answering acyclic CQs with comparison operators between variables ($<, \leq$ etc.) is
harder than answering acyclic CQs with inequalities ($\neq$) since this problem is no longer fixed-parameter tractable.
The query containment problem for CQs with comparisons and inequalities ($\neq, <, \leq$), \ie, whether $Q_1 \subseteq Q_2$, 
has been shown to be $\Pi^p_2$-complete by van der Meyden \cite{vanderMeyden1997}; the effect of several syntactic properties of $Q_1, Q_2$ on the complexity of this problem has been studied by Kolaitis \ea\ \cite{Kolaitis1998}. 
Durand and Grandjean \cite{Durand06} improved Theorem~\ref{thm:PY97_acyclic} from \cite{PY97} by reducing the time complexity by a $log^2|D|$ factor.
Answering queries with views in the presence of comparison operators has been studied by Afrati \ea\ \cite{Afrati2002}. 
Rosati \cite{Rosati07} showed that answering CQs with inequalities is undecidable in description logic.  

\medskip
\textbf{Our Contributions.~~} In this paper we focus on the combined complexity of 
answering CQs with inequalities ($\neq$) where we explore both the structure of the query and the inequalities.
Let $q$ be a CQ with a set of variables, 
$\mI$ be a set of inequalities of the form $x_i \neq x_j$, and 
 $k$ be the number of variables that appear in one of the inequalities in $\mI$ ($k < |q|$). 
We will use $(q, \mI)$ to denote $q$ with inequalities $\mI$, and $D$ to denote the database instance.
We will refer to the combined complexity in $|D|, |q|, k$ by default (and not the data complexity on $|D|$)
unless mentioned otherwise.


The main result in this paper says that 
any query plan for evaluating a CQ can be converted to a query plan for evaluating
the same CQ with arbitrary inequalities, and the increase in running time 
is a factor that only depends on the query:
\begin{theorem}[\textbf{Main Theorem}]\label{thm:main}
Let $q$ be a CQ that can be evaluated in time $T(|q|,|D|)$ using a Select-Project-Join  (SPJ) query plan $\mP_q$. Then, a query plan $\mP_{q, \mI}$ for $(q, \mathcal{I})$ can be obtained 
to evaluate $(q, \mathcal{I})$ in time $g(q, \mI) \cdot \max(T(|q|, |D|),~ |D|)$ for a function $g$ that is independent of the input database.~\footnote{Some queries like $q() = R(x) S(y)$ can be evaluated in constant time whereas to evaluate the inequality constraints we need to scan the relations in $D$.}
\end{theorem}

The key techniques used to prove the above theorem (Sections~\ref{sec:main} and \ref{sec:any-CQ}), 
and our other contributions in this paper (Sections~\ref{sec:gen-YP}, \ref{sec:treewidth-results}, and
\ref{sec:other}) are summarized below.

\begin{enumerate}[leftmargin=0cm,rightmargin=0cm]
\item \textbf{(Section~\ref{sec:main},~\ref{sec:any-CQ})}~
Our main technical contribution is a new {\em projection} operator, called $\mH$-projection. While the standard projection in relational algebra removes all other attributes for each tuple in the output, the new operator computes and retains a certain representation of the group of input tuples that contribute to each tuple in the output.
This representation is of size independent of the database and allows the updated query plan to still correctly filter out certain tuples that do not satisfy the inequalities. In Section~\ref{sec:main} we present the basic algorithmic components of this operator. In Section~\ref{sec:any-CQ}, we show how to apply this operator to transform the given query plan to another query plan that incorporates the added inequalities.  
\item  \textbf{(Section~\ref{sec:gen-YP})}~ We generalize~\autoref{thm:PY97_acyclic} to arbitrary CQs 
 with inequalities  (\ie, not necessarily acyclic) 
	by a simple application of the color-coding technique. 
	In particular, we show (\textbf{\autoref{thm:any_cq_color}}) that any algorithm that computes a CQ $q$ on a database $D$ in time $T(|q|, |D|)$
	can be extended to an algorithm that can evaluate 
	$(q, \mI)$ in time $f(k) \cdot \log(|D|) \cdot T(|q|,|D|)$.
	While Theorem~\ref{thm:main} and Theorem~\ref{thm:any_cq_color} appear similar, 
	there are several advantages of using our algorithm over the color-coding-based technique
	which we also discuss in Section~\ref{sec:gen-YP}.
\item  \textbf{(Section~\ref{sec:treewidth-results})}~
The multiplicative factors dependent on the query in Theorem~\ref{thm:PY97_acyclic}, Theorem~\ref{thm:any_cq_color}, and (in the worst case) Theorem~\ref{thm:main} are exponential in $k$. 
In Section~\ref{sec:treewidth-results} we investigate the combined structure of the queries and inequalities that allow or forbid poly-time combined complexity. We show that, even if $q$ and $\mI$ have a simple structure, answering $(q,\mI)$ can be NP-hard in $k$ (\textbf{Proposition~\ref{thm:np-hard-acyclic}}). We also present a connection with the list coloring problem that allows us to answer certain pairings of queries with inequalities in poly-time combined complexity (\textbf{Proposition~\ref{prop:list-color}}).
\item \textbf{(Section~\ref{sec:other})}
We provide a sufficient condition for CQs, \emph{bounded fractional vertex cover}, that ensures poly-time combined complexity when evaluated with \emph{any set of inequalities $\mI$}.  Moreover, we show that families of CQs with unbounded integer vertex cover are NP-hard to evaluate in $k$ (\textbf{Theorem~\ref{thm:general_hardness}}).
\end{enumerate}


\section{Preliminaries}
\label{sec:prelim}
We are given a CQ $q$, a set of inequalities $\mI$, and a database instance $D$. The goal is to evaluate the query with inequality, denoted by $(q, \mI)$, on $D$.
We will use $vars(q)$ to denote the variables in the body of query $q$ and $\dom$ to denote the active domain of $D$. 
The set of variables in the head of $q$ (\ie, the variables that appear in the output of $q$) is denoted by $head(q)$.
If $head(q) = \emptyset$, $q$ is called a {\em Boolean query}, while if $head(q) = vars(q)$, it is called a {\em full query}.
\par
The set $\mI$ contains inequalities of the form $x_i \neq x_j$, where
$x_i, x_j \in vars(q)$ such that they belong to two distinct relational atoms in the query. 
We do not consider  inequalities of the form $x_i \neq c$ for some constant $c$,  or of the form $x_i \neq x_j$ where $x_i, x_j$ only belong to the same relational atoms because
these can be preprocessed by scanning the database instance and filtering out the tuples that violate these inequalities in time $O(|\mI||D|)$. 
We will use $k$ to denote the number of variables appearing in $\mI$ ($k \leq |vars(q)| < |q|$).\\

\noindent
\textbf{Query Graph, Inequality Graph, and Augmented Graph.~~}
Given a CQ $q$ and a set of inequalities $\mI$, we define three  undirected graphs on $vars(q)$ as the set of vertices:
\par
The \emph{query incidence graph} or simply the \emph{query graph},  denoted by $G^q$, of a query $q$ contains all the variables and the relational atoms in the
query as vertices; an edge exists between a variable $x$ and an atom $S$ if and only if $x$ appears in $S$.
\par
The \emph{inequality graph} $G^{\mI}$ adds an edge between $x_i, x_j \in vars(q)$ if the inequality $x_i \neq x_j$ belongs to $\mI$.
\par
The query $(q, \mI)$ can be viewed as an augmentation of $q$ with additional predicates, where for each inequality $x_i \neq x_j$ we add 
a relational atom $I_{ij}(x_i, x_j)$ to the query $q$, and add new relations $I_{ij}$ to $D$ instantiated to tuples $(a, b) \in \dom \times \dom$ such that $a \neq b$.
The \emph{augmented graph} $G^{q, \mI}$ is the query incidence graph of this augmented query.
Note that $G^{q, \mI}$ includes the edges from $G^q$, and for every edge $(x_i, x_j) \in G^{\mI}$, it includes
two edges $(x_i, I_{ij}), (x_j, I_{ij})$; examples can be found in Section~\ref{sec:treewidth-results}.\\

\noindent
\textbf{Treewidth and Acyclicity of a Query.~~}
We briefly review the definition of the treewidth of a graph and a query.  

\begin{definition}[Treewidth]\label{def:td}
A \emph{tree decomposition} \cite{Robertson1984} of a graph $G(V, E)$ is a tree $T = (I,F)$, with a set $X(u) \subseteq V$
associated with each vertex $u \in I$ of the tree, such that the following conditions are satisfied:
\vspace{-0.2cm}
\begin{enumerate}
\itemsep0em
	\item For each $v \in V$, there is a $u \in I$ such that $v \in X(u)$,
  \item For all edges $(v_1, v_2) \in E$, there is a $u \in I$ with $v_1, v_2\in X(u)$,
  \item For each $v \in V$, the set $\{u \in I~:~  v \in X(u)\}$ induces a connected subtree of $T$.
\end{enumerate}
\vspace{-0.2cm}
The width of the tree decomposition $T = (I, F)$ is $\max_{u \in I} |X(u)| -1$. The \emph{treewidth} of $G$ is the
width of the tree decomposition of $G$ having the minimum width. 
\end{definition}
Chekuri and Rajaraman defined the \emph{treewidth of a query} $q$ as the treewidth of the query incidence graph $G^q$ \cite{ChekuriR00}.
A query can be viewed as a \emph{hypergraph} 
where every hyperedge corresponds to an atom in the query and comprises the variables as vertices that 
belong to the relational atom. The  \emph{GYO-reduction} \cite{Graham79, YO79}
of a query repeatedly removes \emph{ears} from the query hypergraph (hyperedges having at least one variable that does not belong to any other hyperedge) until no further ears exist.  A query is \emph{acyclic} if its GYO-reduction is the empty hypergraph, otherwise it is cyclic. For example, the query 
$P^k(~) = R_1(x_1, x_2), R_2(x_2, x_3), \dots, R_k(x_k, x_{k+1})$ is acyclic, whereas
the query $C^k(~) = R_1(x_1, x_2), R_2(x_2, x_3), \dots, R_k(x_{k}, x_{1})$ is cyclic. 
\par
There is another notion of width of a query called \emph{querywidth} $qw$ defined in terms of \emph{query decomposition} such that the decomposition tree
has relational atoms from the query instead of variables \cite{ChekuriR00}; The relation between the querywidth $qw$ and treewidth $tw$
of a query is given by the inequality $tw/a \leq qw \leq tw+1$, where $a$ is the maximum arity of an atom in $q$.
A query is acyclic if and only if its querywidth is 1; the treewidth of an acyclic query can be $> 1$ \cite{ChekuriR00}. 
The notion of \emph{hypertreewidth} has been defined
by Gottlob \ea\ in \cite{Gottlob1999}. A query can be evaluated in poly-time combined complexity if its treewidth, querywidth, or hypertreewidth is bounded 
\cite{Yan81, ChekuriR00, Gottlob1999, Kolaitis1998, Flum2002}. 


\cut{
There is another notion of width of a query called \emph{querywidth} $qw$ defined in terms of \emph{query decomposition} such that the decomposition tree
has relational atoms from the query instead of variables \cite{ChekuriR00}; The relation between the querywidth $qw$ and treewidth $tw$
of a query is given by the inequality $tw/a \leq qw \leq tw+1$, where $a$ is the maximum arity of an atom in $q$.
A query is acyclic if and only if its querywidth is 1; the treewidth of an acyclic query can be $> 1$ \cite{ChekuriR00}. 
The notion of \emph{hypertreewidth} has been defined
by Gottlob \ea\ in \cite{Gottlob1999}. A query can be evaluated in poly-time combined complexity if its treewidth, querywidth, or hypertreewidth is bounded 
\cite{Yan81, ChekuriR00, Gottlob1999, Kolaitis1998, Flum2002}. 
We have chosen to use tree decomposition and treewidth of 
a query (as opposed to query decomposition and querywidth) as we 
will investigate the structure of the input query and inequalities 
as graphs later in Section~\ref{sec:treewidth-results}.
}

\section{Main Techniques}
\label{sec:main}

In this section, we present the main techniques used to prove Theorem~\ref{thm:main} 
with the help of a simple query $q_2$
that computes the cross product of two relations and projects onto the empty set. 
In particular, we consider the query $(q_2, \mI)$ with an arbitrary set of inequalities $\mI$, where
$q_2(~) = R(x_1, \dots, x_m), S(y_1, \dots, y_{\ell}).$
A na\"{\i}ve way to evaluate the query $(q_2, \mI)$ is to iterate over all pairs of tuples from $R$ and $S$, and check if any such pair satisfies the inequalities in $\mI$. This algorithm runs in time $O(m \ell |R| |S|)$. We will show instead how to evaluate $(q_2, \mI)$ in time $f(q_2, \mI) (|R|+|S|)$ for some function $f$ that is independent of the relations $R$ and $S$.

The key idea is to compress the information that we need from $R$ to evaluate the inequalities by computing a representation $R'$ of $R$ of such that the size of $R'$ only depends on $\mI$ and not on $R$. Further, we must be able to compute $R'$ in time $O(f'(\mI) |R|)$. Then, instead of iterating over the pairs of tuples from $R, S$, we can iterate over the pairs from $R'$ and $S$, which can be done in time $f''(q_2, \mI) |S|$. 
The challenge is to show that such a representation $R'$ exists and that we can compute it efficiently.

We now formalize the above intuition.
Let $X = \{x_1, \cdots, x_m\}$, $Y = \{y_1, \cdots, y_{\ell}\}$. Let $\mH = G^{\mI}$ denote the inequality graph;
since $q_2$ has only two relations, $\mH$ is a bipartite graph on $X$ and $Y$.
If a tuple $t$ from $S$ satisfies the inequalities in $\mI$ when paired with at least one tuple in $R$, we say that $t$ is $\mH$-accepted by $R$, and it contributes to the answer of $(q_2, \mI)$.
For a variable $x_i$ and a tuple $t$, let $t[x_i]$ 
denotes the value of the attribute of $t$ that corresponds to variable $x_i$.
%
%
\begin{definition}[$\mH$-accepted Tuples]
Let $\mH = (X,Y,E)$ be a bipartite graph. We say that a tuple $t$ over $Y$ is {\em $\mH$-accepted} by a relation $R$  if there exists some tuple $t_R \in R$ such that for every $(x_i,y_j) \in E$, we have $t_R[x_i] \neq t[y_j]$.
\end{definition}
Notice that $(q_2, \mI)$ is true if and only if there exists a tuple 
$t_S \in S$ that is $\mH$-accepted by $R$. 

\begin{example}[Running Example]\label{eg:running}
Let us define $\mH_0 = (X,Y,E)$ with $X = \set{x_1, x_2}$, $Y = \set{y_1, y_2, y_3}$ and 
$ E = \set{(x_1, y_1), (x_1, y_2), (x_2, y_2), (x_2, y_3)}$
(see Figure~\ref{fig:example}(a)) 
and consider the instance for $R$ as depicted in~\autoref{fig:example}(b). This setting will be used as our running example.

Observe that the tuple $t = (2,1,3)$ is $\mH_0$-accepted by $R$. Indeed consider the tuple $t' = (3,2)$ in $R$: it is easy to check that all inequalities are satisfied by $t,t'$. In contrast, the tuple $(2,1,2)$ is not $\mH_0$-accepted by $R$.
\end{example}

\begin{figure}[tb]
  \centering
  \subfloat[The bipartite graph $\mH_0$ ]{ 
  \begin{tabular}[b]{c}
     \includegraphics[width=0.25\textwidth]{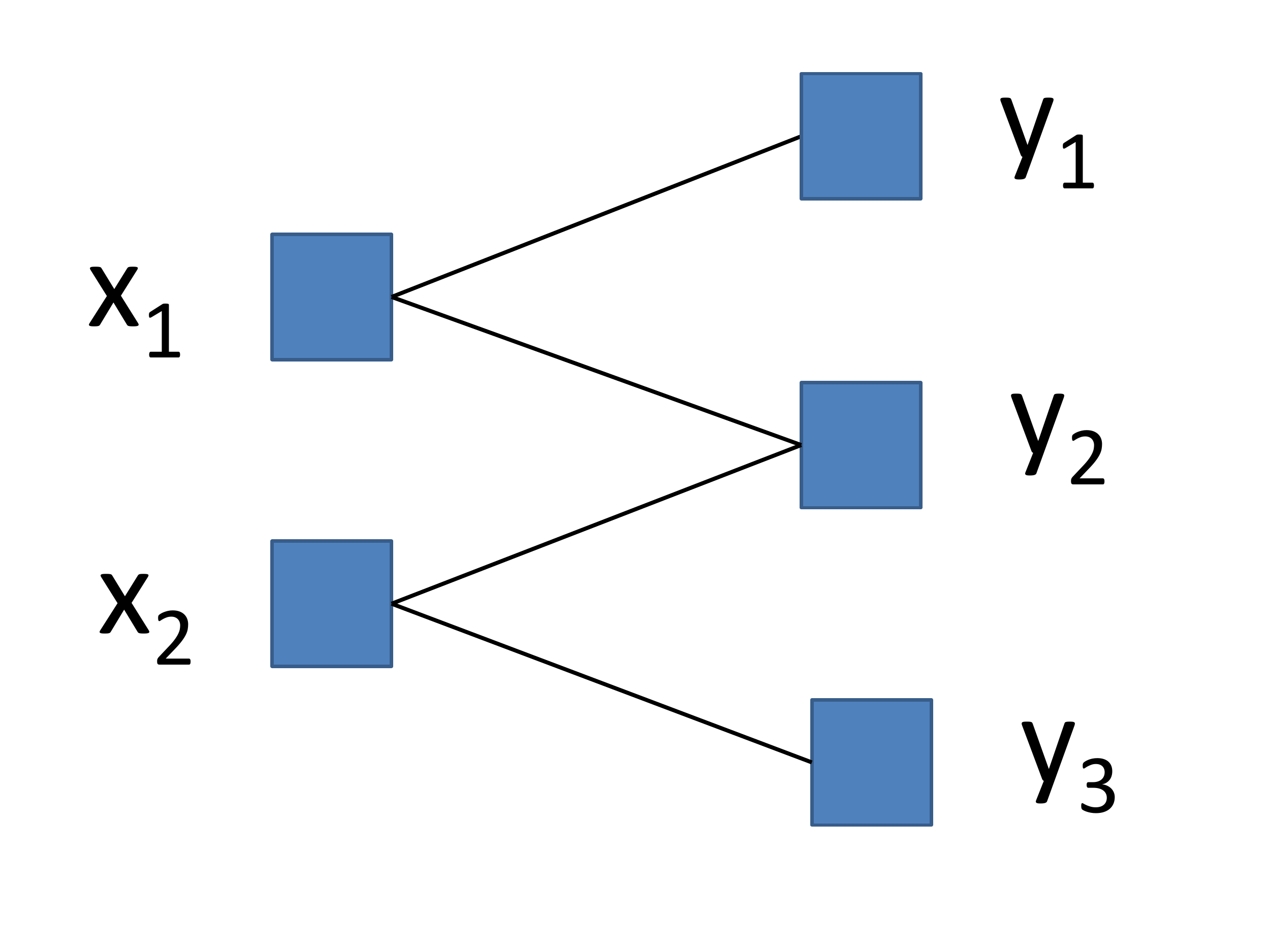}
     \end{tabular}
  }
  \qquad \qquad
  \subfloat[The instance of $R(x_1, x_2)$]{
  \begin{tabular}[b]{ccccc}
   $R =$ & $ \{(1,1)$, & $(1,2)$, &  $(1,4)$, & $(1,8)$, \\
              & $(2,1)$,  & $(2,2)$, & $(2,3)$, & $(2,4)$,\\
              & $(3,2)$, &  $(5,2)$, &  $(10,2)  \}$ &  \\ \\ \\
\end{tabular}}

\caption{{The running example (Example~\ref{eg:running}) for Section~\ref{sec:main}.}}
\vspace{-0.5cm}
\label{fig:example}
\end{figure}

\begin{definition}[$\mH$-Equivalence]
Let $\mathcal{H} = (X,Y,E)$ be a bipartite graph. Two relations $R_1, R_2$ of arity $m = |X|$ are {\em $\mH$-equivalent} if for any tuple $t$ of arity $\ell = |Y|$, the tuple $t$ is $\mH$-accepted by $R_1$  if and only if $t$ is $\mH$-accepted by $R_2$. 
\end{definition}

$\mH$-equivalent relations form an equivalence class comprising instances of the same arity $m$. 
The main result in this section shows that for a given $R$, an $\mH$-equivalent instance $R' \subseteq R$ of size independent of $R$ can be efficiently constructed.
%
%
%
\begin{theorem}\label{lem:main-lem} 
Let $\mathcal{H} = (X,Y,E)$ be a bipartite graph ($|Y| = \ell$) and $R$ be a relation of arity $m = |X|$.
Let $\phi(\mH) =  \ell!  \prod_{j \in [\ell]} d_{\mH}(y_j) $, where $d_{\mH}(v)$ is the degree of a vertex $v$ in $\mH$.
There exists an instance $R' \subseteq R$ such that:
\begin{packed_enum} 
\item $R'$ is  $\mH$-equivalent with $R$
\item  $|R'| \leq e \cdot \phi(\mH)$ 
\item $R'$ can be computed in time $O(\phi(\mH) |R|)$.
\end{packed_enum}
\end{theorem}

To describe how the algorithm that constructs $R'$ works, we need to introduce another notion that describes the tuples of arity $\ell$ that are \emph{not} $\mH$-accepted by $R$.
Let $\bot$ be a value that does not appear in the active domain $\dom$.
%
%
\begin{definition}[$\mH$-Forbidden Tuples]
Let $\mathcal{H} = (X,Y,E)$ be a bipartite graph and $R$ be a relation of arity $m = |X|$. A tuple $t$ over $Y$ with values in $\dom \cup \set{\bot}$ is
{\em $\mH$-forbidden} for $R$ if for any tuple $t_R \in R$ there exist $y_j \in Y$ and $(x_i,y_j) \in E$ such that $t[y_j] = t_R[x_i]$.
\end{definition}

\begin{example}[Continued]
The reader can verify from \autoref{fig:example} that tuples of the form $(1,2,x)$, where $x$ can be any value, are $\mH_0$-forbidden for $R$. Furthermore, notice that the tuple $(1,2,\bot)$ is also $\mH_0$-forbidden (in our construction $(1,2,\bot)$ being $\mH_0$-forbidden implies that any tuple of the form $(1,2,x)$ is $\mH_0$-forbidden). 
\end{example}

Next we formalize the intuition of the above example. We say that a tuple $t_1$ defined over  $Y$ {\em subsumes} another tuple $t_2$ defined over $Y$ if for any $y_j \in Y$, either $t_1[y_j] = \bot$ or $ t_1[y_j] = t_2[y_j]$. 
Observe that if $t_1$ subsumes $t_2$ and $t_1$ is $\mH$-forbidden, $t_2$ must be $\mH$-forbidden as well. A tuple is {\em minimally $\mH$-forbidden} if it is $\mH$-forbidden and is not subsumed by any other $\mH$-forbidden tuple. In our example, $(1,2,1)$ is subsumed by $(1,2,\bot)$, so it is not minimally $\mH_0$-forbidden, but  the tuple $(1,2,\bot)$ is.
Lemma~\ref{lem:semi-main-lem} stated below will be used to prove~\autoref{lem:main-lem}:
\begin{lemma}\label{lem:semi-main-lem}
Let $\mH = (X,Y,E)$ be a bipartite graph, and $R$ be a relation defined on $X$. 
Then, the set of all minimally $\mH$-forbidden tuples of $R$ has size at most  $\phi(\mH) = \ell!  \prod_{j \in [\ell]} d_{\mH}(y_j) $ and it can be computed in time $O(\phi(\mH) |R|)$.
\end{lemma}

To prove the above lemma, we present an algorithm that encodes all the minimally $\mH$-forbidden tuples of $R$ in a rooted tree $T_{\mH}(R)$. The tree has labels for both the nodes and the edges. More precisely, the label $L(v)$ of some node $v$ is either a tuple in $R$ or a special symbol $\bot^*$ (only the leaves can have label $\bot^*$), while the label of an edge of the tree is a pair of the form $(y_j, a)$, where $y_j \in Y$ and $a \in \dom$. The labels of the edges are used to construct a set of $\mH$-forbidden tuples that includes the set of all minimally $\mH$-forbidden tuples as follows:

\emph{For each leaf node $v$ with label $L(v) = \bot^*$, let $(y_{j_1}, a_{j_1}), \dots, (y_{j_m}, a_{j_m})$ be the  edge labels in the order they appear from the root to the leaf. Then, the tuple $\tup(v)$ defined on $Y$ as follows
is an $\mH$-forbidden tuple (but not necessarily minimally $\mH$-forbidden)}:
$$ \tup(v)[y_j] = \begin{cases}
   a_j & \text{if } j \in \set{j_1, \dots, j_m} \\
   \bot       & \text{otherwise } 
\end{cases}$$

\begin{figure}[t]
\centering
\resizebox{0.9\textwidth}{!}{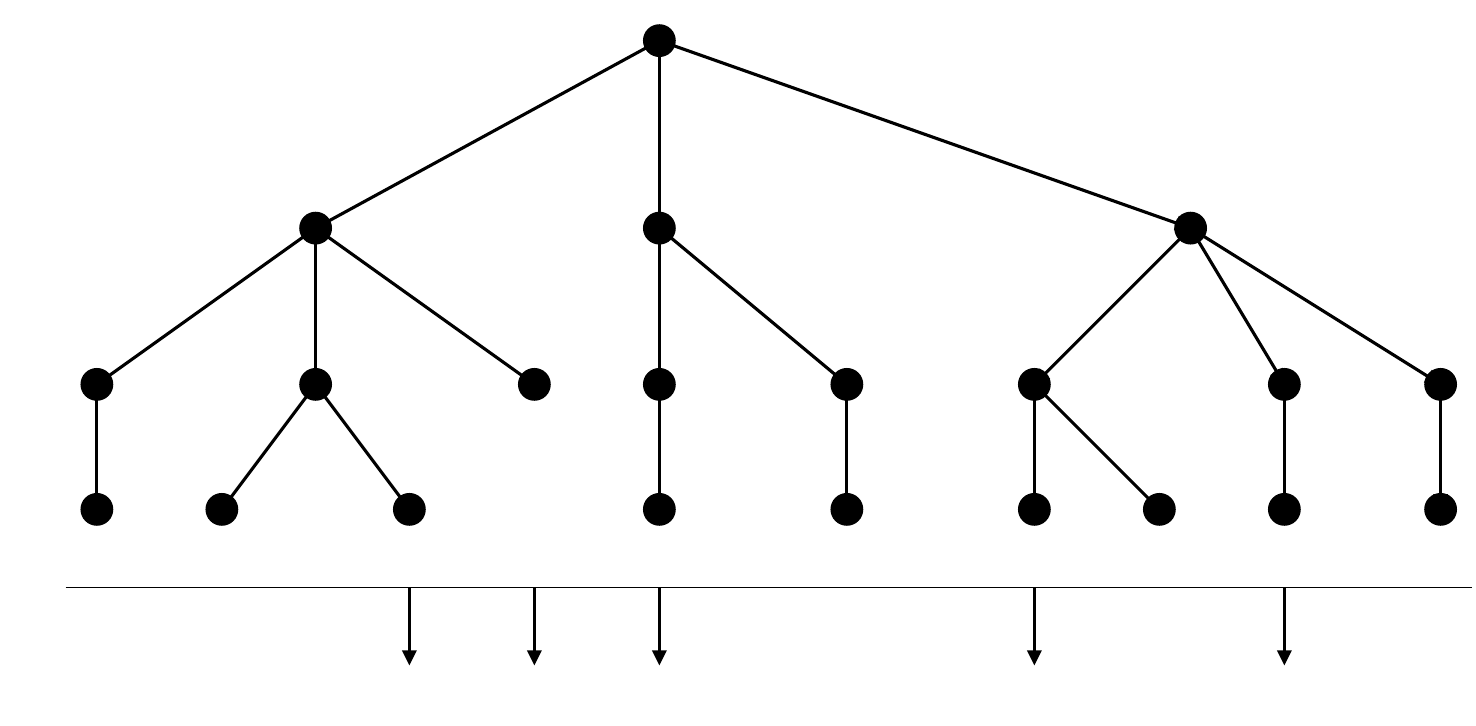}
\caption{The tree $T_{\mH}(R)$ of the running example. The diagram also presents how the $\mH_0$-forbidden tuples are encoded by the tree.}
\vspace{-0.4cm}
\label{fig:example_tree}
\end{figure}

\noindent \textbf{Construction of $T_{\mH}(R)$.}
We construct $T_{\mH}(R)$ inductively by scanning through the tuples of $R$ in an arbitrary order. 
As we read the next tuple $t$ from $R$, we need to ensure that the $\mH$-forbidden tuples that have been so far encoded by the tree are not $\mH$-accepted by $t$: we achieve this by expanding some of the leaves and adding new edges and nodes to the tree. 
Therefore, after the algorithm has consumed a subset $R'' \subseteq R$, the partially constructed tree will be $T_{\mH}(R'')$.

For the base of the induction, where $R'' = \emptyset$, we define $T_{\mH}(\emptyset)$ as a tree that contains a single node (the root $r$) with label $L(r) = \bot^*$. 

For the inductive step, let $T_{\mH}(R'')$ be the current tree and let $t \in R$ be the next scanned tuple.  
The algorithm processes (in arbitrary order) all the leaf nodes $v$ of the tree with $L(v) = \bot^*$. 
Let $(y_{j_1}, a_{j_1}), \dots, (y_{j_p}, a_{j_p})$ be the edge labels in the order they appear on the path from root $r$ to $v$. 
We distinguish two cases (for tuple $t$ and a fixed leaf node $v$):
\begin{enumerate}
\item There exists $j \in \{j_1, \dots, j_p \}$ and edge $(x_{i}, y_{j}) \in E$ such that $t[x_{i}] = a_{j}$. 
In this case, $\tup(v)$ will be $\mH$-forbidden in $R'' \cup \{t \}$; therefore, nothing needs to be done for this $v$.
\item Otherwise (\ie, there is no such $j$), $\tup(v)$ is not a $\mH$-forbidden tuple for $R'' \cup \{t\}$.
We set $L(v) = t$ (therefore, we never reassign the label of a node that has already been assigned to some tuple in $R$). 
There are two cases:
\begin{enumerate}
\item If $p = \ell$,  we cannot expand further from $v$ (and will not expand in the future because now $L(v) \neq \bot^*$), since all $y_j$-s have been already set. 
\item If $p < \ell$, we expand the tree at node $v$. For every edge $(x_i,y_j) \in E$ such that $j \notin \{j_1, \dots, j_p \}$, we add a fresh  node $v^{i,j}$ with $L(v^{i,j}) = \bot^*$ and an edge $(v,v^{i,j})$ with label $(y_{j}, t[x_i])$. Notice that the tuples $\tup(v^{i,j})$ will be now $\mH$-forbidden in $R'' \cup \{t\}$.
 \end{enumerate}
\end{enumerate}
The algorithm stops when either (a) all the tuples from $R$ are scanned or (b) there exists no leaf node with label $\bot^*$. 

\begin{example}[Continued]
We now illustrate the steps of the algorithm through the running example. After reading the first tuple, $t_1 = (1,1)$, the algorithm expands the root node $r$ to three children (for $y_1, y_2, y_3$), labels $L(r) = (1, 1)$ and labels the new edges as $(y_1, 1), (y_2, 1), (y_3, 1)$ and the new three leaves as $\bot^*$.

Suppose the second tuple $t_2 = (1,2)$ is read next. 
First consider the leaf node with label $\bot^*$ that is reached from the root through the edge $(y_1,1)$. At this point, the node represents the tuple $(1,\bot,\bot)$. Observe that are in case (1) of the algorithm, and so the node is not expanded ($t_2[x_1] = 1$ and $m=1 < 3 = \ell$). 
Consider now the third leaf node with label $\bot^*$, reached through the edge $(y_3,1)$. We are now in case (2), and we have to expand the node. The available edges (since we have already assigned a value to $y_3$) are $(x_1, y_1), (x_1, y_2), (x_2, y_2)$. Hence, the node is labeled  $(1,2)$, and expands into three children, one for each of the above edges. These edges are labeled by $(y_1, 1), (y_2, 1), (y_2, 2)$ respectively; then the algorithm continues and at the end the tree in Figure~\ref{fig:example_tree}
is obtained.
\end{example}

The $\mH$-forbidden tuples encoded by the tree are not necessarily minimally $\mH$-forbidden. 
However, for every minimally $\mH$-forbidden tuple there exists a node in the tree that encodes it.
In the running example, we find only two minimally $\mH_0$-forbidden tuples for $R$: $(1,2,\bot)$ and $(2,1,2)$. Furthermore, the constructed tree is not unique for $R$ and depends on the order in which the tuples in $R$ are scanned. The following lemma sums up the properties of the tree construction, and directly implies Lemma~\ref{lem:semi-main-lem}.

\begin{lemma}\label{lem:tree-props}
$T_{\mH}(R)$ satisfies the following properties:
\begin{enumerate}
\item The number of leaves is at most $\phi(\mH) = \ell!  \prod_{j \in [\ell]} d_{\mH}(y_j)$. 
\item Every leaf of $T_{\mH}(R)$ with label $\bot^*$ encodes a $\mH$-forbidden tuple.
\item Every minimally $\mH$-forbidden tuple is encoded by some leaf of the tree with label $\bot^*$.
\end{enumerate}
\end{lemma}

\begin{proof}
We start by showing item (1). The first observation is that the depth of the tree is at most $\ell$. Indeed, consider any path from the root to a leaf, and let $(y_{j_1}, a_{j_1}), \dots, (y_{j_m}, a_{j_m})$ be the labels of the edges. By the construction in step (2), all $j_a$ are pairwise disjoint, and so we can have at most $\ell$ such labels in the path. Notice additionally that each such path visits a subset of the nodes in $Y$ in some order, and maps each node maps it to one of its neighbors in $X$. This implies that the number of leaves in $T_{\mH}(R)$ can be at most $\phi(\mH) = \ell ! \prod_{j \in [\ell]} d_{\mH}(y_j)$. \\

Item (2) is straightforward and follows by the fact that only the expansion step (2) of the algorithm can assign the label $\bot^*$ to a node. \\

Finally, we prove item (3). Let $t$ be a minimally $\mH$-forbidden tuple. We will show that the algorithm will produce $t$ at some leaf of the tree. Our argument will trace $t$ along a path from the root of $T_{\mH}(R)$ to the appropriate leaf.

Consider the tuples of $R$ in the order visited by the algorithm: $t_1, t_2, \dots$. We will show the following inductive statement: for each tuple $t_a$, there exists a leaf node $v_a$ in the tree with label $\bot^*$ such that $\tup(v_a)$ is $\mH$-forbidden for $\{t_1, \dots, t_a\}$ and subsumes $t$.
This statement suffices to prove (3), since at the point where $a = |R| = m$ (\ie, all the tuples in $R$ have been scanned), 
$\tup(v_a)$ must equal $t$ (otherwise $t$ is not minimal), and also $L(v_a) = \bot^*$.

The statement vacuously holds before no tuples from $R$ have been scanned for the root node that encodes $(\bot, \cdots, \bot)$, and forms the basis of the induction. Now, suppose that we are at some tuple $t_a$ and node $v_a$ where the inductive statement holds. Let $t_{a+1}$ be the next tuple in the order. If the algorithm falls into case (1), then $v_{a+1} = v_a$ and $\tup(v_{a+1}) = \tup(v_a)$. Since $\tup(v_a)$ is $\mH$-forbidden for $\set{t_1, \dots, t_a}$ and subsumes $t$, it will be $\mH$-forbidden for $\set{t_1, \dots, t_{a+1}}$ as well, and still subsume $t$. Further, the label of $v_{a+1} = v_{a}$ remains $\bot^*$.

Now suppose we fall into case (2) and $t_{a+1}$ is read. Let $y_{j_1}, \dots, y_{j_p}$ be the variables set so far in $\tup(v_a)$ where $v_a$ is labeled $\bot^*$.
First note that we cannot fall into case (2a), \ie $p < \ell$. 
Indeed, if $p = \ell$ and $t_{a+1}$ satisfies all inequalities with $\tup(v_a)$, then $\tup(v_a)$ is not $\mH$-forbidden.
Since all the positions of $\tup(v_a)$ have been set and $\tup(v_a)$ subsumes $t$, it must hold that $\tup(v_a) = t$.
It follows that $t$ is not $\mH$-forbidden, which contradicts the fact that $t$ remains $\mH$-forbidden after all tuples in $R$ are read.

Therefore, we are in case (2b), and for all $s \in [p]$, there exists some $x_i \in E(\mH)$, $t_{a+1}[x_i] \neq \tup(v_a)[y_{j_b}]$.
When we add $t_{a+1}$, $t$ remains $\mH$-forbidden. Further, 
$\tup(v_a)$ subsumes $t$. 
Therefore there must be some $j \notin \set{j_1, \dots, j_m}$ and $(x_i,y_j) \in E(\mH)$ such that $t_{a+1}[x_i] = t[y_j] = \tup(v_a)[y_j] \neq \bot$. By construction, the algorithm will choose $(x_i, y_j)$ at step (2) to expand $v_a$ and create a child $v_{a+1}$ that connects with an edge $(y_j, t[y_j])$. Note that, $v(t_{a+1})$ still subsumes $t$, is $\mH$-forbidden for the tuples $t_1, \cdots, t_{a+1}$, and has label $\bot^*$, which proves the induction hypothesis for $t_{a+1}$. 
\end{proof}

For our running example, $\phi(\mH_0) = 3! \cdot (1 \cdot 2 \cdot 1) = 12$, whereas the tree $T_{\mH_0}(R)$ has only 10 leaves. We should note here that the bound $\phi(\mH)$ is tight, \ie
there exists an instance for which the number of minimally $\mH$-forbidden tuples is exactly $\phi(\mH)$. 
\footnote{For example, for $\mH_0$ consider the instance $\set{(1,2), (3,4), (5,6)}$. The reader can check that the resulting tree has 12 leaves with label $\bot^*$, and that every leaf leads to a different minimally $\mH$-forbidden tuple.}

We now discuss how we can use the tree $T_{\mH}(R)$ to find a small $\mH$-equivalent relation to $R$. It turns out that the connection is immediate: it suffices to collect the labels of all the nodes (not only leaves) of the tree  $T_{\mH}(R)$ that are not $\bot^*$. More formally:
 \begin{equation}
\mE_{\mH}(R) = \setof{L(v)}{v \in T_{\mH}(R), L(v) \neq \bot^*} \label{equn:equiv}
\end{equation}
We can now show the following result, which completes the proof of \autoref{lem:main-lem}:
\begin{lemma}\label{lem:R_equiv}
The set  $\mE_{\mH}(R)$ is $\mH$-equivalent to $R$ and has size $|\mE_{\mH}(R)| \leq e \cdot  \phi(\mH)$.
\end{lemma}

\begin{proof}
The proof of $\mH$-equivalence is based on the observation that if $T_{\mH}(R) = T_{\mH}(R')$, then $R, R'$ must be $\mH$-equivalent. Indeed, both trees will have the same minimally $\mH$-forbidden tuples, and therefore the set of tuples that are $\mH$-accepted will be same.

To see that $T_{\mH}(R) = T_{\mH}(\mE_{\mH}(R))$, consider $R$ and suppose that we remove some tuple $t$ that does not appear at any label of the tree (and therefore the resulting instance equals $\mE_{\mH}(R)$). If we keep the same order of scanned tuples when constructing both trees, the exact same tree will be produced (since $t$ will not expand any node or add any label).\\

To prove the size bound, we have to give a bound on the number of nodes in the tree, $|V(T_{\mH}(R))|$. For every possible mapping of nodes $y_j$ to one of its neighbors in $\mH$ (there are $\prod_{j \in [\ell]} d_{\mH}(y_j)$ such mappings), consider the subtree of $T_{\mH}(R)$ that contains only the paths from root to leaves where all the edges agree with the mapping (remember that each node creates a child corresponding to an edge $(x_i,y_j)$ of $\mH$); we will first count the nodes of such a subtree. 
This is because the root node can have at most $\ell$ children corresponding to $\leq \ell$ edges in the mapping.
Each child of root can have at most $\ell-1$ children as one of the edges in the mapping has been used in the first level.
Therefore, this subtree will be of size at most 
$$\ell + \ell(\ell-1) + \dots + \ell! = \sum_{i=0}^{\ell} \frac{\ell!}{i!} = \ell! \sum_{i=0}^{\ell} \frac{1}{i!} \leq e \cdot \ell!$$

Since the union of these subtrees will cover all the nodes of $T_{\mH}(R)$, we obtain that  the $e \cdot \phi(\mH)$ is an upper bound for the size of the tree.
\end{proof}

\begin{example}[Continued]
For our running example, the  small $\mH_0$-equivalent relation will be:
$\mE_{\mH_0}(R) = \set{(1,1), (1,2), (1,4), (2,1), (2,3), (3,2), (5,2)}$.
In other words, the tuples $(1,8), (2,2), (2,4), (10,2)$ are redundant and can be removed without 
affecting the answer to the query $(q_2, \mI)$.
\end{example}

Although the set of minimally $\mH$-forbidden tuples is the same irrespective of the order by which the algorithm scans the tuples, the relation $\mE_{\mH}(R)$ depends on this order. It is an open problem to find the smallest possible $\mH$-equivalent relation
for $R$.

\section{Query Plans for Inequalities}
\label{sec:any-CQ}

In this section, we use the techniques presented in the previous section as building blocks
and prove \autoref{thm:main}. 
A \emph{Select-Project-Join (SPJ) query plan} refers to a relational algebra expression that uses only selection ($\sigma$), projection ($\Pi$), and join ($\Join$) operators.
Let $\mP_q$ be any SPJ query plan that computes a CQ $q$ (without inequalities)
on a database instance $D$ in time $T(|q|,|D|)$. 
We will show how to transform $\mP_q$ into a plan $\mP_{q,\mI}$ that computes $(q, \mI)$
in time $g(q, \mI) \cdot \max(T(|q|, |D|)$. 
Without loss of generality, we assume that all the relation names and attributes in the base and derived relations (at intermediate steps in the plan) are distinct.  
Our running example for this section is given below:

\begin{example}\label{ex:running_plan}
Consider the query $(q_0,\mI)$, and the query plan $\mP_{q_0}$ that computes $q_0$:
\begin{align*} 
q_0(w) & = R(x,y,`a`), S(y,z), T(z,w), \quad \mI = \set{x \neq z, y \neq w, x \neq w} \\
\mP_{q_0} & = \Pi_{D} ( \sigma_{E = `a`}(\Pi_{C,E}(R(A,B,E) \Join_{B=B'} S(B',C) )) \Join_{C=C'} T(C',D))
\end{align*}
The query plan $\mP_{q_0} $ is depicted in \autoref{fig:relational_plan}.
\end{example}

\begin{figure}
\centering
\resizebox{1.0\textwidth}{!}{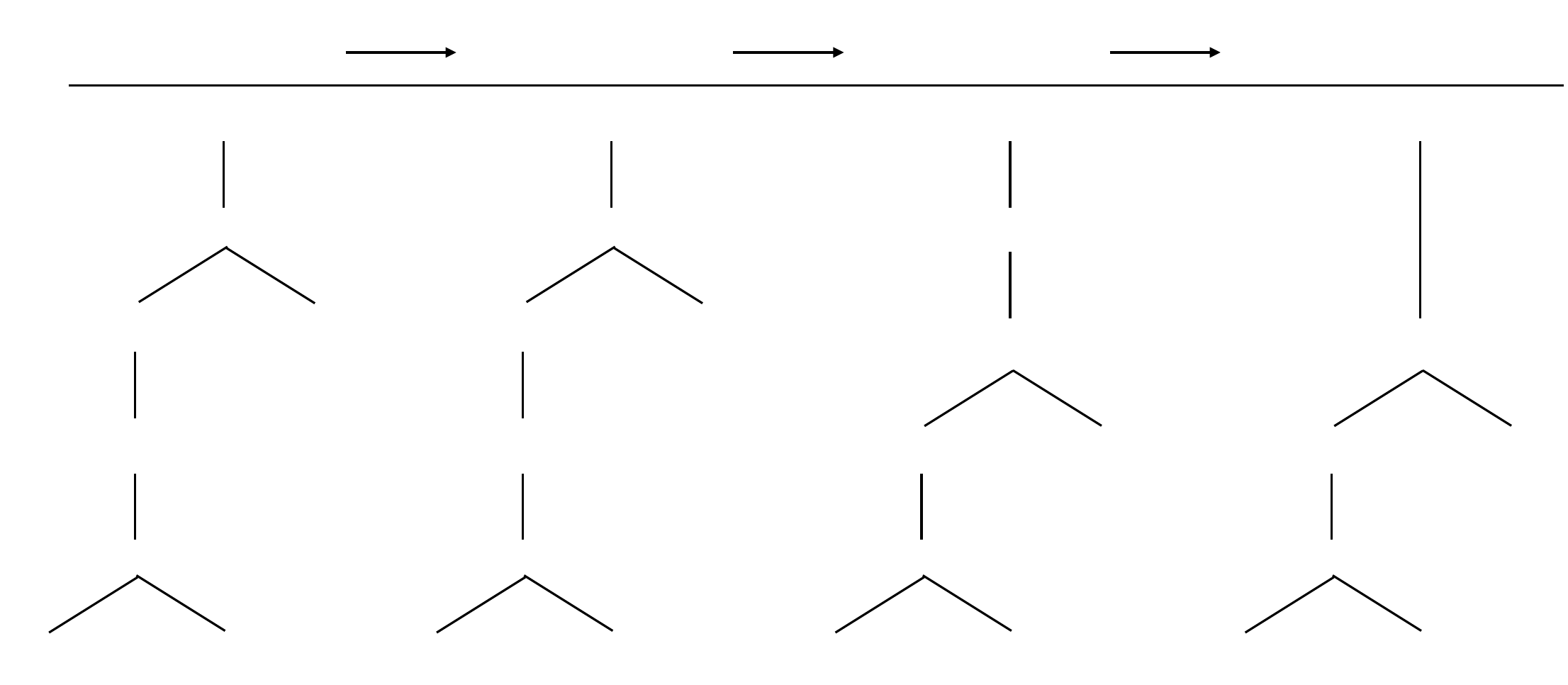}
\caption{The relational plan $\mP_{q_0}$ for Example~\ref{ex:running_plan}, and the transformation to the plan $\mP_{q_0, \top}$.}
\vspace{-0.4cm}
\label{fig:relational_plan}
\end{figure}

Clearly, this plan by itself does not work for $(q_0,\mI)$ as it is losing information that is essential
to evaluate the inequlities, \eg,  $B ( = B')$ is being projected out and it is used later in the inequality 
$x \neq w$ with the attribute $C$ of $T$. 
To overcome this problem while keeping the same structure of the plan, we define a new projection operator that allows us to perform valid algebraic transformations, even in the presence of inequalities. 
Let $\att{R}$ be  the set of attributes that appear in a base or derived relation $R$; a query plan 
or sub-plan $\mP$ is a derived relation with attributes $\att{\mP}$. If $X \subseteq \att{R}$, let $\bar{X}^{R} = \att{R} \setminus X$.

\begin{definition}[$\mH$-Projection]
Let $R$ be a base or a derived relation in $\mP$. Let $X \subseteq \att{R}$ and $\mH = (\bar{X}^R,\att{\mP} \setminus \att{R}, E)$ be a bipartite graph. Then, the $\mH$-projection of $R$ on $X$, denoted $\Pi^{\mH}_{X}(R)$, is defined as
\begin{align}
\Pi^{\mH}_{X}(R) = \bigcup_{\ba \in \Pi_{X}(R)} \mE_{\mH}(\sigma_{X = \ba}(R))
\end{align}
where $\mE_{\mH}$ denotes an $\mH$-equivalent subrelation as defined and constructed in equation ~(\ref{equn:equiv}).
\end{definition}
Intuitively, $\mH$ contains the inequalities between the attributes in $\bar{X}^R$ (that are being projected out) and the attributes of the rest of the query plan. 
The operator $\Pi_X^{\mH}$ first groups the tuples from $R$ according to the values of the $X$-attributes, but then instead of projecting out the values of the attributes in $\bar{X}^R$ for each such group, it computes a small $\mH$-equivalent subrelation according to the graph $\mH$. 

\begin{observation}
The $\mH$-projection of a relation $R$ on $X$ satisfies the following properties:
\begin{enumerate}
\item $\Pi_{X}(R) = \Pi_{X}(\Pi^{\mH}_{X}(R))$
\item $|\Pi^{\mH}_{X}(R)| \leq e \cdot \phi(\mH) \cdot |\Pi_{X}(R)|$ (ref. Lemma~\ref{lem:main-lem})
\end{enumerate}
\end{observation}

\noindent \textbf{First step.~} To construct the plan $\mP_{q,\mI}$ from $\mP_q$, we first create an equivalent query plan $\mP_{q,\top}$ by pulling all the projections in $\mP_q$ to the top of the plan. 
The equivalence of $\mP_q$ and $\mP_{q,\top}$ is maintained by the following 
standard algebraic rules regarding projections:

\begin{description}
\item[(Rule-1) Absorption:] If $X \subseteq Y$, then $\Pi_X (R) = \Pi_X (\Pi_Y (R))$.
\item[(Rule-2) Distribution:] If $X_1 \subseteq \att{R_1}$ and $X_2 = \att{R_2}$, then $\Pi_{X_1 \cup X_2}(R_1 \times R_2) = \Pi_{X_1}(R_1) \times R_2$.
\item[(Rule-3) Commutativity with Selection:] If the selection condition $\theta$ is over a subset of $X$, then  $\sigma_{\theta} (\Pi_X (R)) = \Pi_X (\sigma_{\theta}(R) )$.
\end{description}

\autoref{fig:relational_plan} depicts how each rule is applied in our running example to transform the initial query plan $\mP_{q_0}$ to $\mP_{q_0, \top}$, where the only projection occurs in the top of the query plan. Observe that to distribute a projection over a join $R_1 \Join_{A_1=A_2} R_2$ (and not a cartesian product), we can write it as $\sigma_{A_1 = A_2}(R_1 \times R_2)$, use both (Rule-2) and (Rule-3) to push the projection, and then write it back in the form as $R_1 \Join_{A_1 = A_2} R_2$.
\par
The plan $\mP_{q,\top}$ will be of the form $\mP_{q,\top} = \Pi_X(\mP_0)$, where $\mP_0$ is a query plan that contains only selections and joins. Notice that the plan $\Pi_X (\sigma_{\mI}(\mP_0) )$ correctly computes $(q,\mI)$, since it applies the inequalities before projecting out any attributes.\footnote{From here on we let $\mI$ denote inequalities on attributes and not variables.}
However, the running time is not comparable with that of $\mP_q$ since the structures of the plans 
$\mP_q$ and $\Pi_X (\sigma_{\mI}(\mP_0) )$ are very different. 
 To achieve comparable running time, we modify $\Pi_X (\sigma_{\mI}(\mP_0) )$  by applying the corresponding rules of (Rule-1), (Rule-2), (Rule-3) for $\mH$-projection in the reverse order.\\

\noindent \textbf{Second step.~} To convert projections to $\mH$-projections, first, we replace $\Pi_X$ with $\Pi_X^{\mH_0}$, where $\mH_0 = (\att{\mP_0} \setminus X, \emptyset, \emptyset)$. 
Notice that $\Pi_X^{\mH_0}$ is essentially like $\Pi_X$, but instead of removing the attributes that are not in $X$, the operator keeps an arbitrary witness. Thus, if we compute $\Pi_X^{\mH_0} (\sigma_{\mI}(\mP_0))$, we not only get all tuples $t$ in $(q,\mI)$, but for every such tuple we obtain a tuple $t'$ from $(q^f, \mI)$ such that $t = t'[X]$. 
For our running example, $X = \{D\}$, and therefore, $\mH_0 = (\{A,B,B',C,C',E\}, \emptyset, \emptyset)$
(see the rightmost plan in \autoref{fig:plan_transform}).\\

\noindent \textbf{Third step.~} We next present the 
rules for $\mH$-projections to convert $\Pi_X^{\mH_0} (\sigma_{\mI}(\mP_0))$
to the desired plan $\mP_{q, \mI}$.  To show that the rules are algebraically correct, we need a weaker version of plan equivalence.

\begin{definition}[Plan Equivalence]
Two plans $\mP_1, \mP_2$ are equivalent under $\Pi_X^{\mH}$, denoted $\mP_1 \equiv_X^{\mH} \mP_2 $, if for every tuple $\ba$, $\mE_{\mH}(\sigma_{X=\ba}(\mP_1))$ and $\mE_{\mH}(\sigma_{X=\ba}(\mP_2))$ are $\mH$-equivalent.
\end{definition}

In other words, we do not need to have the same values of the attributes 
that are being projected out by $\Pi_X$
in the small sub-relations $\mE_{\mH}$.
 We write $\mI[X_1,X_2] \subseteq \mI$ to denote 
the inequalities between attributes in subsets $X_1$ and $X_2$.   
 For convenience, we also write $\mI[X,X] = \mI[X]$. We 
 use $E[X_1, X_2]$ in a similar fashion, where $E$ is the set of edges in a bipartite graph. 
Let $\mathbf{A} = \att{\mP_0} $.
We apply the transformation rules for a sub-plan that is of the form $\Pi_X^{\mH} (\sigma_{\mI} (S))$, where $\mI$ is defined on $\att{S}$ and $\mH = (\bar{X}^S, \mathbf{A}, E)$.~\footnote{For the sake 
of simplicity, we do not write the bipartite graph as $\mH = (\bar{X}^S, \mathbf{A} \setminus \att{S}, E)$. However, the transformation rules ensure that the edges $E$ in the bipartite graph are always between
$\bar{X}^S$ and  $\mathbf{A} \setminus \att{S}$.}
The rules are: \\

\noindent \textbf{(Rule-1').} If $X \subseteq Y$ and $\mH' = (\bar{Y}^S, \mathbf{A}, E[\bar{Y}^S, \mathbf{A}])$, then
$$\Pi_X^\mH (\sigma_{\mI}(S))~~ \equiv_X^{\mH}~~ \Pi_X^{\mH} (\Pi_Y^{\mH'} (\sigma_{\mI}(S)))$$ 
In the running example, we have $X = \{D\}$, $Y = \{C,C',D,E\}$, and $\att{S} = \mathbf{A} = \{A, B, B', C, C', D, E\}$. 
The new bipartite graph for Rule-1' in \autoref{fig:plan_transform} (corresponding to Rule-1 in Figure~\ref{fig:relational_plan}) 
is $\mH_1 = (\{A,B,B'\}, \mathbf{A}, \emptyset)$.\\

\noindent \textbf{(Rule-2').} Let $S = R_1 \times R_2$, and $X = X_1 \cup Z_2$, where $X_1 \subseteq \att{R_1} = Z_1$ and $Z_2 = \att{R_2}$. If we define
$\mH'  = (Z_1 \setminus X_1,~ \mathbf{A},~ E[Z_1 \setminus X_1, \mathbf{A}] \cup \mI[Z_1 \setminus X_1, Z_2])$, 
then
$$\Pi_{X_1 \cup Z_2}^\mH (\sigma_{\mI}(R_1 \times R_2))~~ \equiv_X^{\mH}~~ \sigma_{\mI \setminus \mI[Z_1]} (\Pi_{X_1}^{\mH'}(\sigma_{\mI[Z_1]} (R_1)) \times R_2)$$  
This rule adds new edges to the bipartite graph (which is initially empty) from the set of inequalities $\mI$.
In the running example, we have $X_1 = \{C,E\} \subseteq \{A,B,B',C,E\} = Z_1$ and $Z_2 = \{C',D\}$. Since $E(\mH_1) = \emptyset$, to construct the edge set of the new bipartite graph $\mH_2$, we need to find the inequalities that have one attribute in $Z_1 \setminus X_1 = \{A,B,B'\}$ and the other in $Z_2 = \{C',D\}$: these are $A \neq D$ and $B \neq D$. Hence, $\mH_2 = (\{A,B,B'\}, \mathbf{A}, \{(A,D), (B,D)\})$, and the application of the rule is depicted in \autoref{fig:plan_transform}. \\

\noindent \textbf{(Rule-3').} If $\theta$ is defined over a subset of $X$, and $S = \sigma_{\theta}(R)$:
$$\Pi_X^\mH (\sigma_{\mI}(\sigma_{\theta}(R) ))~~ \equiv_X^{\mH}~~ \sigma_{\theta} (\Pi_X^\mH (\sigma_{\mI} (R)))  $$
In the running example, we move the selection operator $\sigma_{E=`a`}$ before the projection operator $\Pi_{C,E}^{\mH_2}$ as the last step of the transformation.

\begin{figure}
\centering
\resizebox{1.0\textwidth}{!}{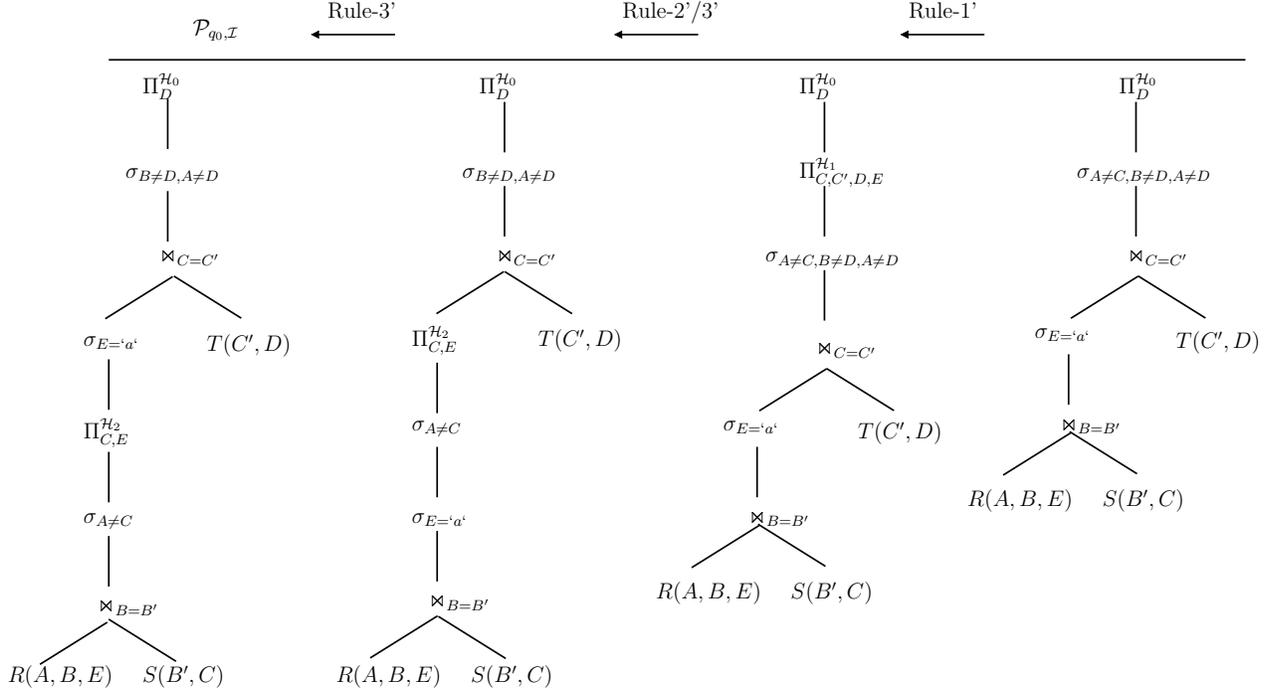}
\caption{The reverse application of rules for Example~\ref{ex:running_plan}. The bipartite graphs defined have edge sets $E(\mH_0) = \emptyset$, $E(\mH_1) = \emptyset$ and $E(\mH_2) = \{(A,D), (B,D)\}$.}
\vspace{-0.3cm}
\label{fig:plan_transform}
\end{figure}

\begin{lemma}
(Rule-1'), (Rule-2'), (Rule-3') preserve the equivalence of the plans under $\Pi_X^{\mH}$.
\end{lemma}

\begin{proof} 
We show the equivalence for each rule.

\introparagraph{(Rule-1')} Denote $S' = \sigma_{\mI}(S)$. It suffices to show that for every tuple $\ba$, $\mE_1 = \mE_{\mH}(\sigma_{X=\ba}(S'))$ and $\mE_2 = \mE_{\mH}(\sigma_{X=\ba}(\Pi_Y^{\mH'}(S')))$ are $\mH$-equivalent. Fix some $X = \ba$.

The one direction is based on the observation that $\Pi_Y^{\mH'}(S') \subseteq S'$. Hence, $\sigma_{X=\ba}(\Pi_Y^{\mH'}(S')) \subseteq \sigma_{X=\ba}(S')$, which implies that if a tuple is $\mH$-accepted by $\mE_2$, it is accepted by $\mE_1$ as well. 

For the other direction, suppose that $t$ is $\mH$-accepted by $\mE_1$. Then, there exists some $s \in \sigma_{X = \ba} (S')$ such that $E \models s \circ t$.\footnote{$s \circ t$  denotes the concatenation of $s,t$.}
Since $\bar{Y}^S \subseteq \bar{X}^S$, $E[\bar{Y}^S, A] \models s \circ t$ and $t$ must be $\mH'$-accepted by $\sigma_{Y = s[Y]}(\sigma_{X = \ba} (S'))$, and consequently by  $\mE_{\mH'}(\sigma_{Y = s[Y]}(\sigma_{X = \ba} (S')))$ as well. 
Then, there exists some $s' \in \mE_{\mH'}(\sigma_{Y = s[Y]}(\sigma_{X = \ba} (S')))$ such that $E[\bar{Y}^S, A] \models t \circ s'$. However, since $s'[Y] = s[Y]$, we must also have that $E \models t \circ s'$. Since $s' \in \Pi_Y^{\mH'}(\sigma_{X = \ba} (S'))$, we conclude that $t$ is $\mH$-accepted by  $\mE_2$.

\introparagraph{(Rule-2')} Denote $R_1' = \sigma_{\mI[Z_1]}(R_1)$ and $\mI_1 = \mI \setminus \mI[Z_1]$. It suffices to show that for every tuple $\ba$, the following are $\mH$-equivalent:
\begin{align*}
\mE_1 & = \mE_{\mH}(\sigma_{X=\ba}(\sigma_{\mI'} (R_1' \times R_2))), \\
\mE_2 & = \mE_{\mH}(\sigma_{X=\ba}(\sigma_{\mI'} (\Pi_{X_1}^{\mH_1} (R_1') \times R_2 )))
\end{align*}
The one direction of the equivalence is based on the fact that $\Pi_{X_1}^{\mH_1}(R_1') \subseteq R_1'$. The other direction is more involved.

Suppose that $t$ is $\mH$-accepted by $\mE_1$. Then, there exists some $s \in R_1' \times R_2$ such that $E, \mI' \models s \circ t$ and $s[X] = \ba$.
Now, consider the tuple $t \circ s[Z_2]$.  The crucial observation is that $t \circ s[Z_2]$ is $\mH_1$-accepted by $\sigma_{X_1 = \ba[X_1]} (R_1')$, and thus by $\mE_{\mH_1}(\sigma_{X_1 = \ba[X_1]} (R_1'))$ as well. 
Then, there exists some $s_1 \in \mE_{\mH_1}(\sigma_{X_1 = \ba[X_1]} (R_1'))$ such that $E  \models t \circ s_1 \circ s[Z_2]$. Finally, observe that the tuple $s' = s_1 \circ s[Z_2]$ belongs in $\Pi_{X_1}^{\mH_1} (R_1') \times R_2$, has $s'[X] = \ba$, and also satisfies $\mI'$. This implies that $t$ is $\mH$-accepted by $\mE_2$.

\introparagraph{(Rule-3')} This is immediate, since the selection $\theta$ is applied only on the attributes in $X$, which are not projected out. 
\end{proof}

After applying the above transformations in the reverse order, the following lemma holds:

\begin{lemma}
Let $\mP_q$ be an SPJ plan for $q$. For a set of inequalities $\mI$, the transformed plan $\mP_{q,\mI}$ has the following properties:
\begin{enumerate}
\item If $\mP_{q,\top} = \Pi_X(\mP_0)$,  the plan $\Pi_{X} (\mP_{q, \mI})$ computes $(q, \mI)$
(after projecting out the attributes that served as witness from $\mP_{q, \mI}$).
\item For every $\Pi_X$ operator in $\mP_q$, there exists a corresponding $\Pi_X^{\mH}$ operator in $\mP_{q,\mI}$ for some appropriately constructed $\mH$.
\item Every intermediate relation $R$ in $\mP_{q,\mI}$ has size at most $e \cdot \max_{\mH} \{\phi(\mH)\} \cdot |R'|$, where $R'$ is the corresponding intermediate relation in $\mP_q$.
\item If $T(|q|, |D|)$ is the time to evaluate $\mP_q$, the the time to evaluate $\mP_{q,\mI}$ increases by a factor of at most $(e \cdot \max_{\mH} \{ \phi(\mH)\})^2$.\end{enumerate}
\end{lemma}

\autoref{thm:main} directly follows from the above lemma. 
To prove the bound on the running time, we use the fact that each operator (selection, projection or join) 
can be implemented in at most quadratic time in the size of the input (\ie, $T(MN) \leq cM^2T(N)$). Additionally, notice that, if $k$ is the vertex size of the inequality graph, then $\max_{\mH} \{ \phi(\mH)\} \leq k! k^k$. Hence, the running time can increase at most by a factor of $2^{O(k \log k)}$ when inequalities 
are added to the query.
In our running example,  $\phi(\mH_0) = 1$, $\phi(\mH_0) = 1$ and $\phi(\mH_2) = 2$, hence the resulting intermediate relations in will be at most $2e$ times larger than the ones in $\mP_{q_0}$.

The following query with inequalities is an example where our algorithm
gives much better running time than the color-coding-based or treewidth-based techniques
described in the subsequent sections.

\begin{example}
Consider $P^k() = R_1(x_1,x_2), R_2(x_2,x_3), \cdots, R_k(x_k,x_{k+1})$ with
inequalities $\mI = \{x_i \neq x_{i+2} \mid i \in [k-1] \}$. 
Let $\mP$ be the SPJ plan that computes this acyclic query in time $O(k |D|)$ by performing joins from left to right and projecting out the attributes as soon as they join. Then, the plan $\mP_{\mI}$ that is constructed has constant $\max_{\mH}\{\phi(\mH)\}$; thus, $(P^k, \mI)$ can be evaluated in time $O(k |D|)$ as well.
\end{example}

\begin{remark}
In this section we compared the running time of queries with inequalities with SPJ plans that compute the query without the inequalities. However, optimal algorithms that compute  CQs may not use SPJ plans, as  the recent worst-case optimal algorithms in~\cite{NgoPRR12,Veldhuizen14} show.
These algorithms apply to conjunctive 
queries without projections, where any inequality can be applied at the end without affecting 
the asymptotic running time. However, there are cases where nonstandard algorithms for Boolean CQs run faster 
than SPJ algorithms, \eg $q() = R(x_1,x_2),R(x_2,x_3),\dots, R(x_{2k},x_1)$,  
can be computed in time $O(N^{2-1/k})$, where $N=|R|$. 
We show in \autoref{sec:non-SPJ} that our techniques can be applied in this case as well. However, it is an open whether we can use them for any black-box algorithm.
\end{remark}

\section{Color-coding Technique and Generalization of Theorem~\ref{thm:PY97_acyclic}}
\label{sec:gen-YP}

In this section, we will review the color-coding technique from \cite{AlonYZ08} and use it to generalize
\autoref{thm:PY97_acyclic} for arbitrary CQs with inequalities (\ie, not necessarily acyclic queries)\footnote{The
$\log^2(|D|)$ factor in \autoref{thm:PY97_acyclic} is reduced to $\log(|D|)$ in \autoref{thm:any_cq_color}, but this is because
one $\log$ factor was due to sorting the relations in the acyclic
query, and now this hidden in the term $T(|q|, |D|)$.}. 
\begin{theorem}\label{thm:any_cq_color}
Let $q$ be a CQ that can be evaluated in time $T(|q|,|D|)$. Then, $(q, \mathcal{I})$ can be computed in time $2^{O(k \log k)} \cdot log(|D|) \cdot T(|q|, |D|)$
where $k$ is the number of variables in $\mI$. 
\end{theorem}
First, we state the original randomized color-coding technique to describe the intuition:
randomly color each value of the active domain 
by using a hash function $h$, use these colors  to check the inequality constraints, and use the actual values 
to check the equality constraints.
\par
For a CQ $q$, let $q^f$ denote the {\em full query} (without inequalities), where every variable in the body appears in the head of the query $q$. 
For a variable $x_i$ and a tuple $t$, $t[x_i]$ (or simply $t[i]$ where it is clear from the context)
denotes the value of the attribute of $t$ that corresponds to variable $x_i$.
Let $t \in q^f(D)$. We say that $t$ {\em satisfies the inequalities $\mI$}, denoted by $t \models \mI$, 
if for each $x_i \neq x_j$ in $\mathcal{I}$, $t[x_i] \neq t[x_j]$.
We say that $t$ {\em satisfies the inequalities $\mI$ with respect to  the hash function $h$}, denoted by $t \models_h \mI$,
if for each such inequaity $h(t[x_i]) \neq h(t[x_j])$.

Recall that $k$ is the number of variables that appear in $\mI$.
Let $h$ be a perfectly random hash function $h: \dom \rightarrow [p]$ (where $p \geq k$). For any $t \in q^f(D)$ 
if $t$ satisfies $\mI$, then with high probability it also satisfies $\mI$ with respect to $h$, \ie,
$\mathbf{Pr}_h [t \models_h \mI~ |~ t \models \mI]$ $\geq$ $\frac{p (p-1) \cdots (p-k+1)}{p^k}$ 
$\geq$ $e^{-2 \sum_{i=1}^{k-1} (i/p)} \geq e^{-k}$,
%
%
where we used the fact that $1-x \geq e^{-2x}$ for $x \leq \frac{1}{2}$. 
Therefore, by repeating the experiment $2^{O(k)}$ times we can evaluate a Boolean query with constant probability. 

This process can be derandomized leading to a deterministic algorithm (for evaluating any CQ, not necessarily Boolean) 
by selecting
$h$ from a family $\mathcal{F}$ of $k$-perfect hash functions. A $k$-perfect family guarantees that for every tuple of arity at most $k$ (with values from the domain $\dom$), there will be some $h \in \mathcal{F}$ such that for all $i, j \in [k]$,
if $t[i] \neq t[j]$, then  $h(t[i]) \neq h(t[j])$ 
(and thus if $t \models \mI$, then $t \models_h \mI$) It is known (see~\cite{AlonYZ08}) that we can construct a $k$-perfect family of size $|\mathcal{F}| = 2^{O(k)} \log(|\dom|) = 2^{O(k)} \log |D|$.
\footnote{Assuming $\dom$ includes only the attributes that appear as variables in the query $q$, $|\dom| \leq |D||q|$.}

A coloring $\mathbf{c}$ of the vertices   of the inequality graph $G^{\mathcal{I}}$ with $k$ colors is called a \emph{valid $k$-coloring}, if for each $x_i \neq x_j$ we have that $c_i \neq c_j$ where $c_i$ denotes the color of variable $x_i$ under $\mathbf{c}$.  
Let $\mathcal{C}(G^{\mI})$ denote all the valid colorings of $G^{\mI}$. 
For each such coloring $\mathbf{c}$ and any given hash function $h: \dom \rightarrow [k]$, we can define a subinstance $D[\mathbf{c},h] \subseteq D$ such that for each relation $R$, 
$R^{D[\mathbf{c},h]} = \setof{t \in R^D}{\forall x_i \in vars(R), h(t[x_i]) = c_i} $. 
In other words, the subinstance $D[\mathbf{c},h]$ picks only the tuples that under the hash function $h$ agree with the coloring $\mathbf{c}$ of the inequality graph. 
Then the algorithm can be stated as follows:
\begin{itemize}
	\item \textbf{Deterministic Algorithm:} \textit{For every hash function $h: \dom \rightarrow [k]$ 
	in a $k$-perfect hash family $\mathcal F$, for every valid $k$-coloring $\mathbf{c} \in \mathcal{C}(G^{\mI})$ of the variables, 
	evaluate the query $q$ on the sub-instance $D[\mathbf{c}, h]$.
	Output $\bigcup_{h \in \mathcal{F}} \bigcup_{\textbf{c} \in \mathcal{C}(G^{\mI})} q(D[\mathbf{c}, h])$.}
\end{itemize}

\begin{proof}[Proof of Theorem~\ref{thm:any_cq_color}]
Suppose 
$$q^{(h)}(D) = \setof{t[head(q)]}{t \in q^f(D) \models_h \mI }$$ 
Then the union $\bigcup_{h \in \mathcal{F}} q^{(h)}(D)$ produces the result of the query
(this is because for any tuple $t \in q^f(D) $, there exists a hash function $h \in \mathcal F$ that satisfies
all the inequalities in $\mI$).
 In the rest of this subsection, we will show how to compute $q^{(h)}(D)$ for a fixed hash function $h: \dom \rightarrow [p]$,
 $p \geq k$, 
using the coloring technique in time bounded by $2^{O(k \log k)} T(|q|, |D|)$. 

Let $\mathbf{C}$ be a \emph{valid $p$-coloring} of the vertices of the inequality graph $G^{\mathcal{I}}$, such that whenever $x_i \neq x_j$, we have that $c_i \neq c_j$ where $c_i$ denotes the color of variable $x_i$ under $\mathbf{c}$. For each such coloring, we can define a subinstance $D[\mathbf{C},h] \subseteq D$ such that for each relation $R$, 
$$R^{D[\mathbf{C},h]} = \setof{t \in R^D}{\forall x_i \in vars(R), h(t[x_i]) = c_i} $$ 
In other words, the subinstance $D[\mathbf{C},h]$ picks only the tuples that under the hash function $h$ agree with the coloring $\mathbf{C}$ of the inequality graph. 

\begin{lemma}\label{lem:color}
Let $\mathcal{C}(G^{\mI})$ denote all the valid colorings of $G^{\mI}$. Then, 
$$q^{(h)}(D) = \bigcup_{\mathbf{C} \in \mathcal{C}(G^{\mI})} q(D[\mathbf{C},h])$$
\end{lemma}

\begin{proof}
Let $t \in q^f(D)$ $\models_h \mI$. Let $\mathbf{C}$ be the coloring such that for every $x_i \in V(G^{\mI})$, we set $c_i = h(t[x_i])$. We will show that $\mathbf{C}$ is a valid coloring of $G^{\mI}$. Indeed, if $x_i \neq x_j \in \mI$, it must be that $h(t[x_i]) \neq h([t_j])$ (Since $t \models_h \mI$) and hence $c_i \neq c_j$. Thus, we have that $t[head(q)] \in q(D[\mathbf{C},h])$.

For the other direction, let $t \in q^f(D[\mathbf{C},h])$ for a valid coloring $\mathbf{C}$. For any inequality $x_i \neq x_j$, we will have $h(t[x_i]) = c_i \neq c_j = h(t[x_j])$, and hence $t \models_h \mI$.
\end{proof}

The algorithm now iterates over all hash functions $h \in \mathcal{F}$, and all valid colorings of $G^{\mathcal{I}}$ with $p$ colors, and for each combination computes $q(D[\mathbf{c}],h)$. The output result is:
$$ \bigcup_{h \in \mathcal{F}, \mathbf{C} \in \mathcal{C}(G^{\mI})} q(D[\mathbf{C},h])$$
The running time is $O(|\mathcal{F}| \cdot |\mathcal{C}(G^{\mI})| \cdot T(q, |D|))$. As we discussed before $|\mathcal{F}| \leq 2^{O(p)} \log |D|$ and $|\mathcal{C}(G^{\mI})| \leq k^p$. \autoref{thm:any_cq_color} follows by choosing $p=k = |V(G^{\mI})|$ (the smallest possible value of $p$).
\end{proof}

\textbf{Comparison of Theorem~\ref{thm:main} with Theorem~\ref{thm:any_cq_color}.~~}
The factors dependent on the query in these two theorems ($g(q, \mI)$ in Theorem~\ref{thm:main} and $f(k)$ in Theorem~\ref{thm:any_cq_color}) are both bounded by 
$2^{O(k\log k)}$. However, our technique outperforms the color-coding technique in several respects. 
First, the randomized color-coding technique is simple and elegant, but is unsuitable to implement
in a database system that typically aims to find deterministic answers. On the other hand, 
 apart from the additional $\log(|D|)$ factor, 
 the derandomized color-coding technique
demands the construction of a new $k$-perfect hash family 
for every database instance and query, and therefore may not be efficient for practical purposes.
Our algorithm 
requires no preprocessing and can be 
applied in a database system by maintaining the same query plan
and using a more sophisticated projection operation.
More importantly, the color coding technique is {\em oblivious} of the combined structure of the query and the inequalities. As an example, consider the path query $P^k$, together with the inequalities $\mI_1 = \{x_i \neq x_{i+2}: i \in [k-1]\}$. 
The color-coding-based algorithm has a running time of $2^{O(k \log k)} |D| \log |D|$. However, as discussed in Section~\ref{sec:any-CQ}, we can compute this query in time $O(k |D|)$, thus the exponential dependence on $k$ 
is eliminated.

\section{CQs and Inequalities with Polynomial Combined Complexity}\label{sec:treewidth-results}

In this section, we investigate classes of queries and inequalities that entail a poly-time combined complexity
for $(q, \mI)$ in terms of the treewidths of query graph $G^q$,  inequality graph $G^{\mI}$, 
and augmented graph $G^{q, \mI}$. 
If the augmented graph $G^{q, \mI}$ has bounded treewidth, then 
$(q, \mI)$ can be answered in poly-time combined complexity \cite{Yan81, ChekuriR00}. 
We give examples of such $q$ and $\mI$ below:

\begin{figure}[t]
	\centering
	\begin{minipage}{0.7\linewidth}
	\includegraphics[scale=0.32]{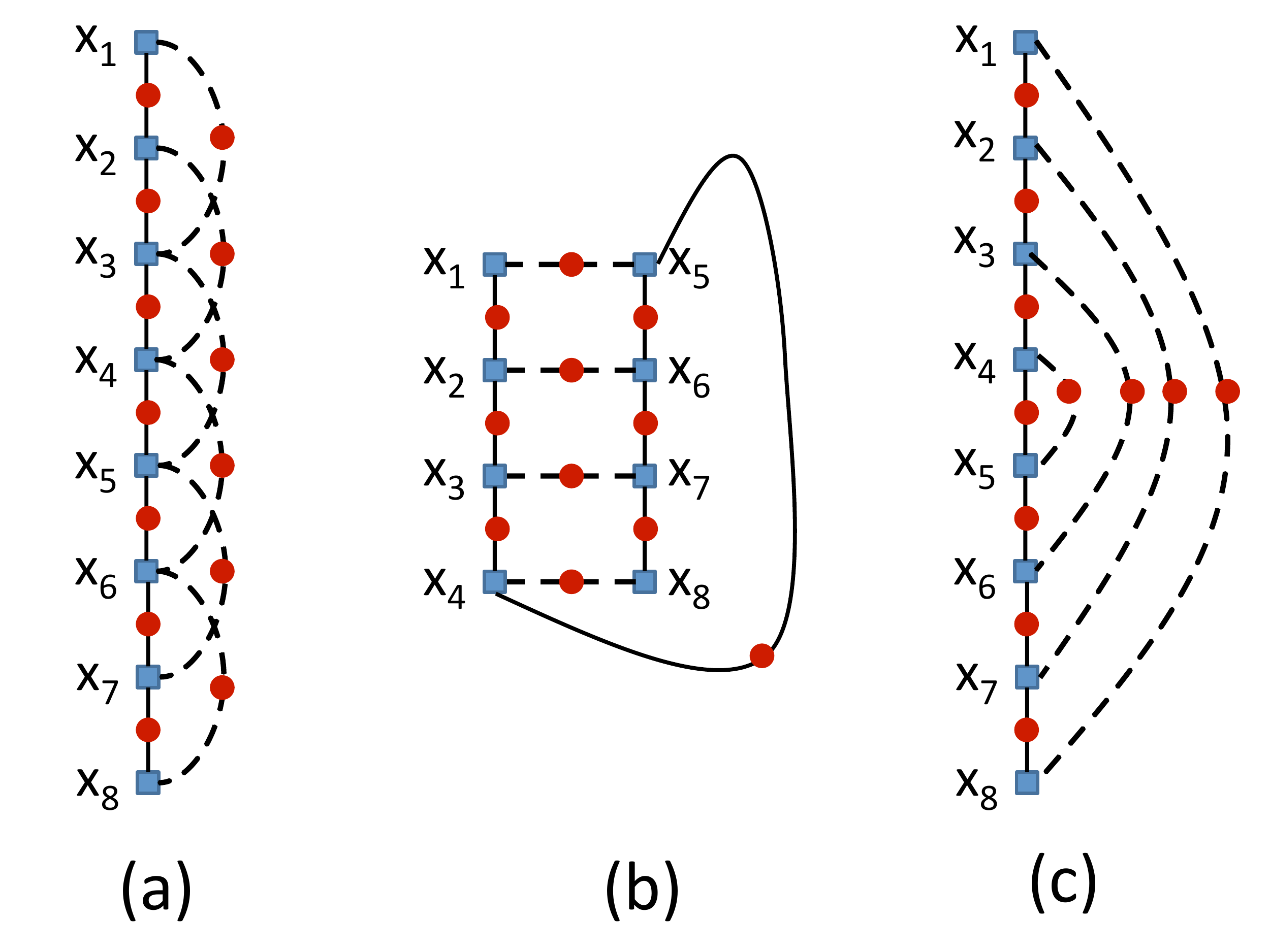}
	\end{minipage}
	\begin{minipage}{0.25\linewidth}
	\includegraphics[scale=0.17]{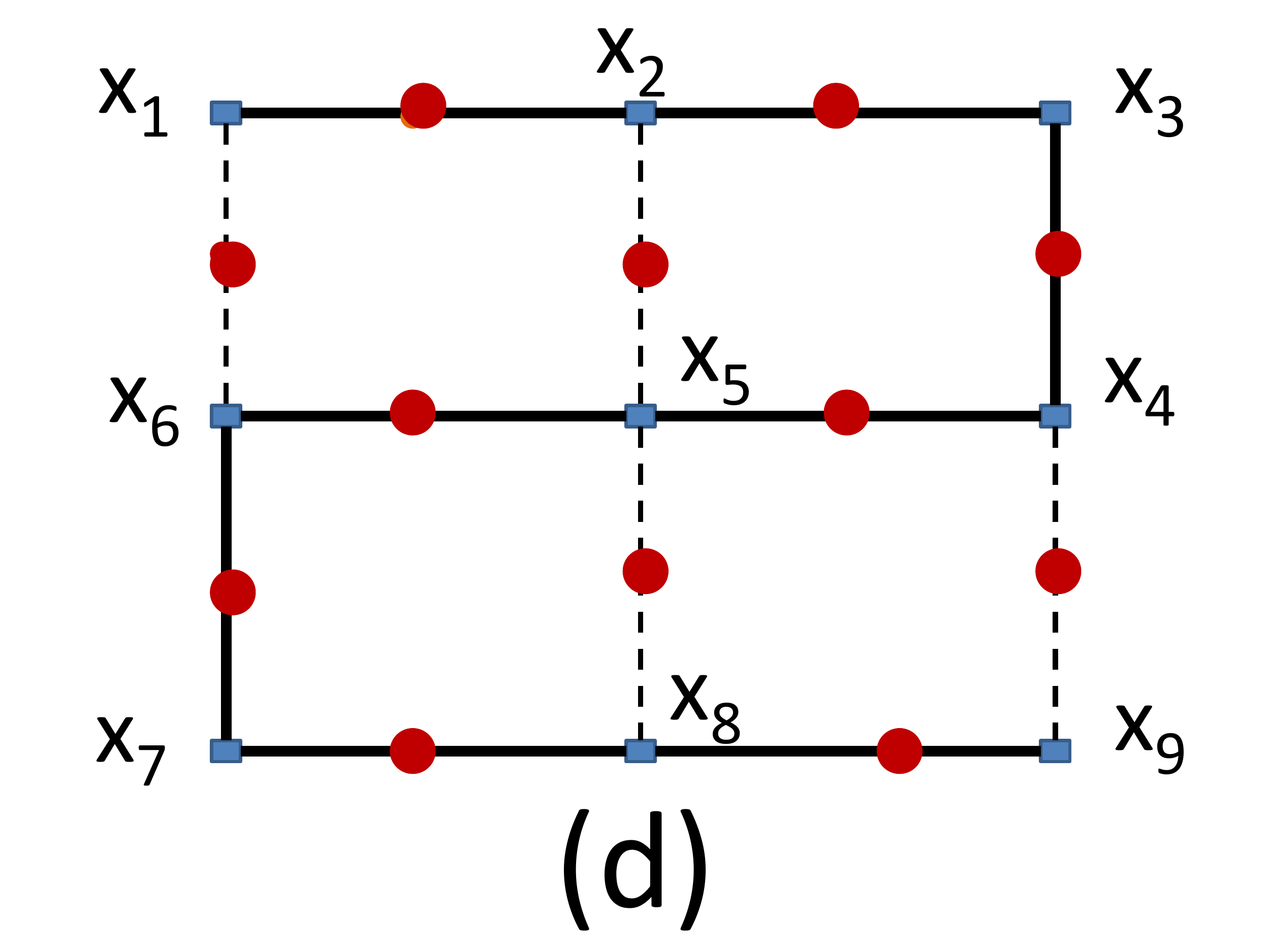}
	\end{minipage}
	\caption{{Augmented graphs 
	for 
	Example~\ref{eg:tw-eg} ($k= 7$)
	and Example~\ref{eg:tw-eg-2} ($k= 8$).
	The solid and dotted edges come from the query and inequalities respectively;
	the 
	blue squares denote variables, and red circles denote (unnamed) relational atoms: 
  (a) $(P^7, \mI_1)$, (b) $(P^7, \mI_2)$, (c) $(P^7, \mI_3)$, (d) $(P^7, \mI_4)$}}
  \vspace{-0.4cm}
	\label{fig:tw-eg}
\end{figure}

\begin{example}\label{eg:tw-eg}
Consider the path query:
$P^k(~) = R_1(x_1, x_2), R_2(x_2, x_3), \dots, R_k(x_{k}, x_{k+1})$, which is acyclic,  
and consider the following inequality patterns (see Figure~\ref{fig:tw-eg}):
\begin{enumerate}
\item $(P^k, \mI_1)$ where $\mI_1 = \{x_i \neq x_{i+2}: i \in [k-1]\}$ has treewidth 2.
 A tree-decomposition with treewidth 2 (\ie, maximum node in the tree
has size 3) is $\set{x_i, x_{i+1}, x_{i+2} }$ --- $\set{ x_{i+1}, x_{i+2}, x_{i+3} }$ --- $\cdots$.
\item $(P^k, \mI_2)$ where $\mI_2 = \{x_i \neq x_{i+\frac{k}{2}}: i \in [\frac{k+1}{2}]\}$ has treewidth 3 ($k$ is odd): The incidence graph without the edge $(x_{\frac{k}{2}}, x_{\frac{k}{2}}+1)$ has the structure of a  $\frac{k}{2} \times 2$ grid and therefore has treewidth 2 (see~\autoref{fig:tw-eg}(b)). We can simply add the node $x_{\frac{k}{2}}$ to all nodes in this tree-decomposition to have a decomposition with treewidth 3. 
\item $(P^k, \mI_3)$ where $\mI_3 = \{x_i \neq x_{k-i+1}: i \in [\frac{k+1}{2}]\}$ has treewidth 2 ($k$ is odd): A tree-decomposition with treewidth 3 can be obtained by going back and forth along the inequality (dotted) edges: $\set{ x_{i+1}, x_i, x_{k-i+1} }$ --- $\set{ x_{k-i+1}, x_{k-i}, x_{i+1} }$ --- $\cdots$. For example, in~\autoref{fig:tw-eg}(c) the decomposition can be $\set{x_2, x_1, x_8}$ --- $\set{x_8, x_7, x_2}$ --- $\set{x_3, x_2, x_7}$ --- $\set{x_7, x_6,x_3}$ --- $\set{x_4, x_3, x_6}$ --- $\set{x_6, x_5, x_4}$.
\end{enumerate}
\end{example}

However, for certain inputs our algorithm in Section~\ref{sec:any-CQ}
can outperform the treewidth-based techniques 
since it considers the inequality structure more carefully.
For instance, even though the augmented graph of $(P^k, \mI_1)$ has treewidth 2 (see Figure~\ref{fig:tw-eg} (a)), the techniques of
	\cite{Yan81} will give an algorithm with running time $O(poly(k) |D|^2)$, whereas the algorithm in Section~\ref{sec:any-CQ}
gives a running time of $O(k |D|)$. 	

Indeed, the treewidth of $G^{q, \mI}$ is at least as large as the treewidth of $G^q$ and $G^{\mI}$.
As mentioned earlier, when $G^{\mI}$ is the complete graph on $k+1$ variables (with treewidth $= k+1$), 
answering $(P^k, \mI)$
is as hard as finding if a graph on $k+1$ vertices has a Hamiltonian path, and therefore is NP-hard in $k$.
Interestingly, even when both $G^{q}$ and $G^{\mI}$ have bounded treewidths,
$G^{q, \mI}$ may have unbounded treewidth as illustrated by the following example:
\begin{example}\label{eg:tw-eg-2}
Consider $(P^k, \mI_4)$ (see Figure~\ref{fig:tw-eg}(d)), where $k+1 = p^2$ for some $p$. Algebraically, we can write $\mI_4$ as:
$ \mI_4 = \setof{x_i \neq x_{\lfloor i/p \rfloor +1 + 2p-(i \mod p)}}{i=1, \dots, p(p-1)}$.
	The edges for $P^k$ are depicted in the figure as an alternating path on the grid with solid edges, whereas the remaining edges are dotted and correspond to the inequalities. Here both $G^{P^k}$ and $G^{\mI_4}$ have treewidth 1, but $G^{P^k, \mI_4}$ has treewidth $\Theta(\sqrt{k})$. 
\end{example}
However, this does not show that evaluation of the query $(P^k, \mI_4)$ is NP-hard in $k$,
which we prove below by a reduction from the \emph{list coloring problem}:
\begin{definition}[List Coloring]\label{def:list-color}
Given an undirected graph $G = (V,E)$, and a list of admissible colors $L(v)$ for each vertex $v \in V$,
  list coloring asks whether there exists a coloring $c(v) \in L(v)$ for each vertex $v$ such that 
 the adjacent vertices in $G$ have different colors.
\end{definition}
The list coloring problem generalizes the coloring problem, and therefore is NP-hard.
List coloring is NP-hard even on grid graphs with $4$ colors and where $2 \leq |L(v)| \leq 3$ for each vertex $v$ \cite{Demange2013}; we show NP-hardness for $(P^k, \mI_4)$ by a reduction from list coloring on grids.
\begin{proposition}\label{thm:np-hard-acyclic}
The combined complexity of evaluating $(P^k, \mI_4)$ is NP-hard, where both the query $P^k$ and the inequality graph $G$ are acyclic (have treewidth 1).
\end{proposition}

\begin{proof}
We reduce from list coloring on grid graphs, which is known to be NP-complete
with $c = 4$ colors and where $2 \leq |L(v)| \leq 3$ for each vertex $v$ \cite{Demange2013}. 

Given an instance of the list coloring problem where the graph $G$ is a $p \times p$ grid-graph, we create an instance of $P^k, \mI_4$
as shown in Figure~\ref{fig:tw-eg}(d), where $k+1 = p^2$.
We denote by $x_i$ both the vertices in $G$ as welll as the variables in $P^k$.
For each $i \in [k]$, we create an instance
$$R_i(x_i, x_{i+1}) = \{(a, b)~:~ a \neq b~ \textrm{and}~ a, b \in L(x_i) \times L(x_j)\}$$
The inequalities $\mI_4$ are as shown in the figure: $x_i \neq x_j$. Note that each vertex $v$ in the grid graph $G$
appears in one of the relations so its domain in the query is bounded by $L(v)$.
\par
Suppose the list coloring instance has a valid coloring, \ie, every vertex $v$ in $G$ can be colored $c[v] \in L(v)$ such that 
for each edge $(u, v)$, $c[u] \neq c[v]$. This gives an yes-instance to the query $(P^k, \mI_4)$.
Similarly, if the query has a yes instance, that corresponds to a yes-instance of the list coloring problem.
\end{proof}

In fact, the above proposition can be generalized as follows:
 \emph{if the graph $G^{q,\mI}$ is NP-hard for list coloring for a query $q$
where each relation has arity 2, then evaluation of the query $(q, \mI)$ is also NP-hard in the size of the query}.

On the contrary, 
$(q, \mI)$ may not be hard in terms of combined complexity if the treewidth of $G^{q, \mI}$ is unbounded,
which we also show with the help of the list coloring problem.
Consider the queries $F^k(~)  = R_1(x_1), R_2(x_2), \dots, R_k(x_k)$.
Given inequalities $\mI$, the evaluation of $(F^k, \mI)$ 
 is \emph{equivalent} to the list coloring problem on 
 the graph $G^{\mI}$ when the available colors for each vertex $x_i$ are the tuples in $R_i(x_i)$.
Since list coloring is NP-hard: 
\begin{proposition}\label{prop:cross-product-nphard}
The evaluation of  $(F^k, \mI)$ is NP-hard in $k$ for arbitrary inequalities $\mI$.
\end{proposition}

Therefore, answering $(F^k, \mI)$ becomes NP-hard in $k$ even for this simple class of queries if we allow arbitrary set of inequalities $\mI$
(this also follows from Theorem~\ref{thm:general_hardness}).
However, list coloring can be solved in polynomial time for certain graphs $G^{\mI}$:
\begin{packed_item}
\item 
 \textbf{Trees} (the problem can be solved in time $O(|V|)$ independent of the available colors\cite{JansenS97}), and in general graphs of constant treewidth.
\item  \textbf{Complete graphs} (by a reduction to \textit{bipartite matching}). \footnote{We can construct a bipartite graph where all vertices $v$ appear on one side, the colors appear on the other side, and there is an edge $(v, c)$ if $c \in L(v)$. Then the list coloring problem on complete graph is solvable if and only if there is a perfect matching in the graph.}
\end{packed_item}
In general, if the connected components of $G$ are either complete graphs or have constant treewidth, list coloring can be solved in polynomial time.
Therefore, on such graphs as $G^{\mI}$, the query $(F^k, \mI)$ can be computed in poly-time in $k$ and $|D|$.
Here we point out that none of the other algorithms given in this paper can give a poly-time algorithm in $k, |D|$ for $(F^k, \mI)$
when $G^{\mI}$ is the complete graph (and therefore has treewidth $k$).
The following proposition generalizes this property:

\begin{proposition}\label{prop:list-color}
Let $q$ be a Boolean CQ, where each relational atom has arity at most 2. If $q$ has a \emph{vertex cover} (a set of variables that can cover all relations in $q$) of constant size 
and the list coloring problem on $G^{\mI}$ can be solved in poly-time, then $(q, \mI)$ can be answered in poly-time combined complexity.
\end{proposition}

\begin{proof}
Let $X  = \{x_{i_1}, \dots, x_{i_c}\}$ be the vertices of the vertex cover. Consider each possible instantiation of these variables from the domain $\dom$; the number of such instantiation is $|\dom|^c$. For each such instantiation $\mathbf{\alpha}$ consider the updated query $q^{\mathbf{\alpha}}$. Since $X$ is a vertex cover and each relation has arity $\leq 2$, in $q^{\mathbf{\alpha}}$ each relation has at most one free variable. Relations with single variable that has been instantiated to a unique constant from $\mathbf{\alpha}$ or relations where both the variables have been instantiated can be evaluated by a linear scan of the instance and removed thereafter.
Similarly, relations with arity 2 where exactly one of the two variables has been instantiated to a constant can be evaluated by removing tuples from the instance that are not consistent with this constant. These steps can be done in poly-time in combined complexity.
In the reduced query, each relation has exactly one free variable and therefore is equivalent to $F^n$ for some $n$ ($n = $ the number of relations in the query where exactly one variable belongs to $X$). Hence if list coloring can be solved in poly-time on $G^{\mI}$, $(q^{\mathbf{\alpha}}, \mI)$ for each instantiation $\alpha$, and therefore $(q, \mI)$ can be solved in poly-time in combined complexity.
\end{proof}

To see an example, consider the star query $Z^n(~) = R_1(y,x_1), \dots, R_n(y, x_n)$ which has a vertex cover $\{y\}$ of size 1. 
We iterate over all possible values of $y$: for each such value $\alpha \in \dom$, the query $R_1(\alpha, x_1), \dots, R_n(\alpha, x_n) $ is equivalent to $F^n$,
and therefore $(Z^n, \mI)$ can be evaluated in poly-time in combined complexity when $G^{\mI}$ is an easy instance of list coloring.

\section{CQs with Polynomial Combined Complexity for All Inequalities}
\label{sec:other}

This section aims to find  CQs $q$ such that computing $(q, \mI)$ has poly-time combined complexity, no matter what the choice of $\mI$ is. Here we present a sufficient condition for this, and a stronger necessary condition.

A {\em fractional edge cover} of a CQ $q$ assigns  a number $v_R$ to each relation $R \in q$ such that for each variable $x$, $\sum_{R: x \in vars(R)} v_R \geq 1$. A {\em fractional vertex packing} (or, \emph{independent set}) of $q$ assigns a number $u_x$ to each variable $x$, such that $\sum_{x \in vars(R)} u_x \leq 1$ for every relation $R \in q$.
By duality, the minimum fractional edge cover is equal to the maximum fractional vertex packing.
When each $v_R \in \{0, 1\}$ we get an \emph{integer edge cover}, and when each $u_x \in \{0, 1\}$ we get an \emph{integer vertex packing}.
\begin{definition}
 A family $\mathcal{Q}$ of Boolean CQs has {\em unbounded fractional (resp. integer) vertex packing} if  there exists a function $T(n)$
 such that 
 for every integer $n>0$ it can output  in time $poly(n)$ 
 a query $q \in \mathcal{Q}$
 that has a fractional (resp. integer) vertex packing of size at least $n$
 (counting relational atoms as well as variables).

A family $\mathcal{Q}$ of Boolean CQs has \emph{bounded fractional (resp. integer) vertex packing} if there exists a constant $b > 0$ such that for any $q \in \mQ$, the size of any fractional (resp. integer) vertex packing is $\leq b$.
\end{definition}

The class of path queries $P^k(~) = R_1(x_1, x_2), R_2(x_2, x_3), \dots, R_k(x_k, x_{k+1})$
and cycle queries $C^k(~) = R_1(x_1, x_2), R_2(x_2, x_3), \dots, R_k(x_k, x_1)$
are examples of classes of unbounded vertex packing.

The main theorem of this section is stated below:

\begin{theorem}\label{thm:general_hardness}
The following hold:
\begin{enumerate}
	\item If a family of Boolean CQs $\mQ$ has unbounded integer vertex packing,  the combined complexity of $(q, \mI)$ for $q \in \mathcal{Q}$ is NP-hard.
	\item If a family of CQs $\mQ$ has bounded fractional vertex packing, then for each $q \in \mQ$, $(q, \mI)$ can be evaluated in poly-time combined complexity for any $\mI$.
\end{enumerate} 
\end{theorem}

The NP-hardness in this theorem follows by a reduction from \textsc{3-Coloring}, whereas the poly-time algorithm uses the bound given 
by Atserias-Grohe-Marx \cite{GroheM06, AGM2013} in terms of the size of minimum fractional edge cover of the query,
and the duality between minimum fractional edge cover and maximum fractional vertex packing.

\begin{proof}
\textbf{1. NP-hardness.~~} 
We do a reduction from \textsc{3-Coloring}. Let $G = (V,E)$ be an undirected graph, where $n = |V|$. The goal is to color $G$ with $3$ colors such that no two adjacent vertices have the same color. Given the family $\mQ$ with unbounded integer vertex packing,
we construct a query with inequality $(q, \mI)$ where $q \in \mathcal{Q}$, and an instance $D$ such that $G$ admits a 3-coloring if and only if $(q, \mI)$ is true on the instance $D$. 

Fine in polynomial time the query $q \in \mathcal{Q}$ 
such that $q$ has an integer vertex packing $X$ of size $n$. 
Let $X = \{x_1, \dots, x_n\}$ be the variables in the vertex packing. We create an instance $D$ as follows. Note that each relation $R$ contains at most one variable $x_i$ from $X$. If $R$ contains no such variable, then $R^{D} = \set{(0,0, \dots, 0)}$ (a single tuple with value 0 for all variables). Otherwise, let $x_i \in vars(R)$ and without loss of generality (wlog.), let $x_i$ be at the first position of $R$; then, $R^{D} = \setof{(c, 0,0, \dots)}{c =1,2,3}$. Observe that the size of the instance $D$ is at most $3 \cdot |q|$ and it is constructed in poly-time. 

By construction the answer to the full query $q^f$ of $q$ is $q^f(D) = \set{1,2,3}^n \times \{(0,0,\dots)\}$, where wlog. all $x_1, \dots, x_n$ appear at the first $n$ positions of the head of $q^f$. Therefore, each variable $x_i$, $i=1, \dots, n$ can obtain each color independent of the other variables. Finally, we construct a one-to-one mapping from each vertex $v \in V$ to a unique variable $x_v \in \set{x_1, \dots, x_n}$, and define $\mI = \setof{x_u \neq x_v}{(u,v) \in E}$. 

Now it is easy to verify that $G$ has a valid 3-coloring if and only if $(q, \mI)$ is true on $D$.\\

\textbf{2. Algorithm for queries with bounded fractional vertex packing.~~} Since the maximum fractional vertex packing of any $q \in \mQ$ is $\leq b$, the minimum fractional edge cover is also $\leq b$ by duality. Thus from \cite{GroheM06, AGM2013}, $q^f(D)$ can be evaluated in poly-time in combined complexity (in time $O(|q|^2 |D|^{b+1})$). 
Further, $|q^f(D)| \leq |D|^b$ \cite{GroheM06, AGM2013}. We first compute $q^f(D)$, then for each tuple in $q^f(D)$ we check whether it satisfies the inequalities, and finally apply the projection to get the answers to $(q, \mI)$ in poly-time in combined complexity.
\end{proof}

In this paper, we illustrate the properties with examples. Consider the family $S^k(~) = R(x_1, \dots, x_k)$ for $k \geq 1$: this has vertex packing of size $ = 1$ and therefore can be answered trivially  in poly-time in combined complexity for any inequality pattern $\mI$. On the other hand, the class of path queries $P^k$ 
mentioned earlier has unbounded vertex packing (has a vertex packing of size $\approx \frac{k}{2}$), and therefore for certain set of inequalities (\eg, when $G^{\mI}$ is a complete graph), the query evaluation of $(P^k, \mI)$ is NP-hard in $k$. Similarly, the class $F^k(~)  = R_1(x_1), R_2(x_2), \dots, R_k(x_k)$ mentioned earlier
			has unbounded vertex packing, and is NP-hard in $k$ with certain 
			inequality patterns (see Proposition~\ref{prop:cross-product-nphard}).

\autoref{thm:general_hardness} is not a dichotomy or a characterization of easy CQs w.r.t. inequalities, since there is a gap between the 
maximum fractional and integer vertex packing.\footnote{For example, for the complete graph on $k$ vertices, the maximum integer vertex packing is of size 1 whereas the maximum fractional vertex packing is of size $\frac{k}{2}$.}



%
%

\section{Conclusion}
\label{sec:conclusion}
We studied the complexity of CQs with inequalities and compared the complexity of query answering with and without the inequality constraints.
Several questions remain open: Is there a property that gives a dichotomy of query evaluation with inequalities
both for the class of CQs, and for the class of CQs along with the inequality graphs? 
What can be said about unions of conjunctive queries (UCQ) and recursive datalog programs?
Can our techniques be used as a black-box to extend any algorithm for CQs, \ie, not necessarily based on
SPJ query plans, to evaluate CQs with inequalities?

\bibliographystyle{plain}
\bibliography{bib}

\newpage
\appendix

\section{Comparison of Our Techniques in Section~\ref{sec:main} with Other Related Work}\label{app:related}

Monien in~\cite{Monien85} defines the notion of {\em $q$-representatives} for families of sets. Given a family of sets $F$, where each set has $p$ elements, $\hat F \subseteq F$ is a $q$-representative if for every set $T$ of size $q$, there exists some set $U \in F$ with $U \cap T = \emptyset$ if and only if there exists a set $\hat U \in \hat F$ such that $\hat U \cap T = \emptyset$. Observe that a $q$-representative is a special case of an $\mH$-equivalent relation: indeed, we can model the family $F$ as a relation $R^F$ of arity $p$ (where we do not care about the order of the attributes), and define $\mH$ as the full bipartite graph with edge set $[p] \times [q]$. Then, if we write $\mE_{\mH}(R^F)$ back to a family of sets, it is a $q$-representative of $F$.

Our techniques also generalize the notion of {\em minimum samples} presented in~\cite{Durand06}, which corresponds to $\mH$-forbidden tuples of a relation in the case where $\mH = (X,Y,E)$ has $|X| = |Y| $ and $E(\mH)$ forms a perfect matching between $X$ and $Y$. Several of the definitions and algorithmic ideas were inspired by both~\cite{Durand06, Monien85}.

\section{Computing Cycles with non-SPJ plans}
\label{sec:non-SPJ}

Let $R$ be a binary relation, and define for any $k \geq 1$ the query:
$$ C_{2k}() = R(x_1, x_2), R(x_2, x_3), \dots, R(x_{2k}, x_1)$$
One can think intuitively that $R$ represents the edges of a directed graph $G$.

\begin{theorem}
Let $R$ be of size $N$. The query  $C_{2k}$ can be computed in time $O(N^{2-1/k})$.
\end{theorem}

\begin{proof}
Let $\delta$ be some threshold parameter. We say that a value $a$ is a {\em heavy hitter} if the degree $|\sigma_{X=a} R(X,Y)| \geq \delta$, otherwise it is light. The algorithm distinguishes two cases.

First, we compute all the $2k$-cycles that contain some heavy hitter value. We have at most $N/\delta$ such values. For each such value, we can compute
$$ C_{2k}^{(a)}() = R(a, x_2), R(x_2, x_3), \dots, R(x_{2k}, a)$$
Observe that this is an acyclic query now that $a$ is a fixed value, so we can compute this query in time $O(k N)$. Hence, to compute all possible cycles in this case we need $O(k N^2/\delta)$ time.

Second, we compute whether there exists a cycle $C_{2k}$ that uses only light values. To do this, let $R'$ be the subset of $R$ that contains only the light values.  The maximum degree is $\delta$, so the queries
\begin{align*} 
q_1(x_1, x_k) & = R(x_1, x_2), \dots, R(x_{k-1}, x_k) \\
q_2(x_k, x_1) & = R(x_k, x_{k+1}), \dots, R(x_{2k}, x_1)
\end{align*}
each contain at most $N \delta^{k-1}$ answers, which we can compute in time $O(k N \delta^{k-1})$ by performing consecutive joins. However, $|q_1|, |q_2|$ have size at most $N$, and we can compute their intersection in time $O(N \log N)$. So the total running time for this case is $O(N \delta^{k-1})$.

To balance the two cases, we must have $N^2/\delta = N \delta^{k-1}$ or $\delta = N^{1/k}$.
\end{proof}

We can combine the above algorithm with our technique as follows. Suppose the query now is $(C_{2k}, \mI)$ for some set of inequalities $\mI$. Observe that the first case is easy to handle, since we know how to compute $C_{2k}^{(a)}$ in time $O(kN \times \max_{\mH} (\phi(\mH))^2)$, where $\phi(\mH)$ depends only on the inequality structure. For the second case, instead of computing $q_1, q_2$, we consider the full queries
\begin{align*} 
q_1^f(x_1, x_2, \dots, x_k) & = R(x_1, x_2), \dots, R(x_{k-1}, x_k) \\
q_2^f(x_k, \dots, x_{2k}, x_1) & = R(x_k, x_{k+1}), \dots, R(x_{2k}, x_1)
\end{align*}
and compute $(q_1^f, \mI_1)$, $(q_2^f, \mI_2)$, where $\mI_1$ are the inequalities defined only between the variables of $x_1$ (and similarly for $\mI_2$). Since these queries have size at most $N \delta^{k-1}$, we can compute the full answers and apply the inequalities at the end. To compute the intersection between $q_1, q_2$, let $\mI_{12} = \mI \setminus (\mI_1 \cup \mI_2)$. We then compute $\Pi_{x_1, x_k}^{\mH_1} (q_1^f)$, where $\mH_1 = ( \{x_2, \dots, x_{k-1}\}, \{x_{k+1}, \dots, x_{2k} \}, \mI_{12})$, and similarly $\Pi_{x_1, x_k}^{\mH_2} (q_2^f)$ with a symmetrically defined $\mH_2$. The resulting projections have size at most $N \cdot \phi(\mH_i)$ for $i=1,2$, so we can then compute their intersection in time $O(N \log N)$ and then apply the inequalities in $\mI_{12}$.

\end{document}